%% file: BooleanSolve.tex
\documentclass[11pt]{article}
\usepackage[numbers]{natbib}
\usepackage[latin1]{inputenc} 
\usepackage{amsmath,amsthm}
\usepackage{amssymb}
\usepackage{times}
\usepackage{algorithmicx}
\usepackage{algorithm}
\usepackage{algpseudocode}
\usepackage{a4wide} 
\usepackage[paperwidth=8.5in,paperheight=11in,margin=1in]{geometry}
\usepackage{tikz}
\usepackage{pgfplots}
\usepackage{authblk}
\usepackage[small,compact]{titlesec}
\usepackage{mathptmx}
\usepackage{graphics}

\newtheorem{defn}{Definition}
\newtheorem{thm}{Theorem}
\newtheorem{cor}{Corollary}
\newtheorem{conj}{Conjecture}

\newtheorem{prop}{Proposition}
\newtheorem{lem}{Lemma}

\DeclareMathOperator{\dreg}{d_{reg}}
\DeclareMathOperator{\dwit}{d_{wit}}
\DeclareMathOperator{\LM}{LM}

\begin{document}

\title{On the Complexity of Solving Quadratic Boolean Systems}
\author[1]{Magali Bardet}
\author[2,3,4]{Jean-Charles Faug\`ere}
\author[5]{Bruno Salvy}
\author[2,3,4]{Pierre-Jean Spaenlehauer}
\affil[1]{\'Equipe Combinatoire et Algorithmes -- Universit\'e de Rouen/LITIS}
\affil[2]{INRIA, Paris-Rocquencourt Center, SALSA Project}
\affil[3]{CNRS, UMR 7606, LIP6}
\affil[4]{UPMC, Univ Paris 06, LIP6
       {\small UFR Ing\'enierie 919, LIP6,
       Case 169, 4, Place Jussieu, F-75252 Paris}}
\affil[5]{INRIA, Paris-Rocquencourt Center, Algorithms Project}
\date{}
\maketitle

\begin{abstract}
A fundamental problem in computer science is to find all the common zeroes of $m$ quadratic polynomials in $n$ unknowns over $\mathbb{F}_2$. The cryptanalysis of several modern ciphers reduces to this problem. Up to now, the best complexity bound was reached by an exhaustive search in $4\log_2 n\,2^n$ operations. We give an algorithm that reduces the problem to a combination of exhaustive search and sparse linear algebra. This algorithm has several variants depending on the method used for the linear algebra step. 
Under precise algebraic assumptions on the input system, we show that the deterministic variant of our algorithm has complexity bounded by $O(2^{0.841n})$ when $m=n$, while a probabilistic variant of the Las Vegas type has expected complexity $O(2^{0.792n})$.
Experiments on random systems show that the algebraic assumptions are satisfied with probability very close to~1. 
We also give a rough estimate for the actual threshold between our method and exhaustive search, which is as low as~200, and thus very relevant for cryptographic applications.
\end{abstract}

\section{Introduction}
\input{intro}

\section{Algorithm}
\label{sec:algo}
\noindent\textbf{Notations.}
\input{notations}

\subsection{Macaulay matrix}
\input{macaulay}

\subsection{Algorithm}
\input{algorithm}

\subsection{Testing Consistency of Sparse Linear Systems}
\input{wiedemann}

\section{Complexity Analysis}
\label{sec:compl}
Algorithm {\sf BooleanSolve} deals with a large number of Macaulay matrices in degree~$d_0$. We first obtain bounds on the row and column dimensions of Macaulay matrices, as well as their number of nonzero entries, in terms of the degree. We then 
bound the witness degree by~$d_0$. The complexity analysis is concluded by optimizing the value of the ratio $k/n$ that governs the number of variables evaluated in the first exhaustive search.

\subsection{Sizes of Macaulay Matrices}
\input{sizemacaulay}

\subsection{Bound on the Witness Degree of Inconsistent Systems}
\input{linearity-new2}

\subsection{Complexity}
\input{compl_alpha}

\input{picture}

\section{Extensions and Applications}
\label{sec:appli}
\subsection{Adding Redundancy to Avoid Rank Defects }
\label{sec:proba}
\input{proba}

\subsection{Cases with Low Degree of Regularity}
\input{low_degree}

\input{crypto}

\subsection*{Acknowledgments}
We wish to thank D.~Bernstein, C.~Diem, E.~Kaltofen and L.~Perret for
valuable comments and pointers to important references.  This work was
supported in part by the Microsoft Research-INRIA Joint Centre and by
the CAC grant (ANR-09-JCJCJ-0064-01) and the HPAC grant of the French
National Research Agency.

\bibliographystyle{abbrv}
\bibliography{biblio}
\end{document}

%% file: intro.tex
\textbf{Motivation and Problem Statement.}  Solving multivariate
quadratic  polynomial systems is a fundamental
problem in Information Theory. Moreover, \emph{random}
instances seem difficult to solve. Consequently, the security of
several multivariate cryptosystems relies on its hardness,
either directly (e.g., HFE \citep{Pat96}, UOV
\citep{KipPatGou99},\ldots) or indirectly (e.g., McEliece~\citep{EuroCrypt10}).
In some cases, systems of special types have to be solved, but recent proposals like the new Polly Cracker type
cryptosystem \citep{Asia2011} rely on the hardness of solving \emph{random}
systems of equations. This motivates the study of the complexity
of generic polynomial systems. A particularly important case for applications in cryptology is the Boolean case; in that case both the coefficients
and the solutions of the system are over \(\mathbb{F}_{2}\). The main problem to be solved is the following: 

\begin{quote}
{\bf The Boolean Multivariate Quadratic Polynomial Problem (Boolean MQ)}\\
{\bf Input:} \((f_{1},\ldots,f_{m})\in\mathbb{F}_2[x_{1},\ldots,x_{n}]^m\) with \(\deg(f_{i})=2\) for $i=1,\dots,m$.\\
{\bf Question:} Find -- if any -- \emph{all} \(z\in \mathbb{F}_{2}^{n}\) such that \(f_{1}(z)=\cdots=f_{m}(z)=0\).
\end{quote}
Another related problem stems from the fact that in many cryptographic applications, 
it is sufficient to find \emph{at least one} solution of the corresponding polynomial system (in that case a solution is the original clear message or is related to the secret key). For instance, the stream cipher QUAD~\citep{QUADEuro,QUADJSC} relies on the iteration of a set of multivariate quadratic polynomials over  \(\mathbb{F}_{2}\)
so that the security of the keystream generation is related to the difficulty  of finding at least one solution of the Boolean MQ problem. Thus, we also consider the following variant of the
Boolean MQ\ problem:
\begin{quote}
{\bf The Boolean Multivariate Quadratic Polynomial Satisfiability
Problem (Boolean MQ
SAT)}\\
{\bf Input:} \((f_{1},\ldots,f_{m})\in\mathbb{F}_2[x_{1},\ldots,x_{n}]^m\) with \(\deg(f_{i})=2\) for $i=1,\dots,m$.\\
{\bf Question:} Find -- if any -- \emph{one} \(z\in \mathbb{F}_{2}^{n}\) such that \(f_{1}(z)=\cdots=f_{m}(z)=0\).
\end{quote}
Testing for the existence of a solution is an NP-complete problem (it is plainly in NP and 3-SAT can be reduced to it~\citep{FraenkelYesha1979}). 
Clearly, the Boolean MQ\ problem is at least as hard as Boolean MQ SAT, while an exponential complexity is achieved by exhaustive search. 

Throughout this paper, \emph{random} means distributed  according to the uniform distribution (given $m$ and $n$, a random quadratic polynomial is uniformly distributed if all its coefficients are independently and uniformly distributed over $\mathbb{F}_2$).
The relation between the difficulties of  Boolean MQ and Boolean MQ SAT depends on the relative values of~$m$ and~$n$.
When \(m> n\), the number of solutions of the algebraic system is 0 or 1 with large
probability and thus finding one or all solutions is very similar, while when~$m=n$, the probability that a random system has at least one solution
over \(\mathbb{F}_{2}\)
tends to \(1-\frac{1}{e}\approx 0.63\) for large~$n$~\citep{FusBac07}. 
Hence if we have to find a
least one solution of a system with \(m<n\) equations in \(n\)
variables it is enough to specialize \(n-m\) variables randomly
in \(\mathbb{F}_{2}\); the resulting system has at least one
 solution with limit probability \(0.63\) and is easier to solve (since
the number of equations and variables is only \(m\)). Consequently,
in the remainder of this article we restrict ourselves to the
case \(m\geq n\).

To the best of our knowledge, in the \emph{worst case}, the best complexity bound to solve the Boolean MQ problem is obtained by a modified exhaustive search in $4\log_2(n)2^n$ operations~\citep{CHES10}.
Being able to decrease significantly this complexity is a long-standing
open problem and is the main goal of this article.  It is
crucial for practical applications to have sharp estimates of the
asymptotic complexity: it is especially important in the cryptographic
context where this value may have a strong impact on the sizes of the
keys needed to reach a given level of security.

\medskip

\noindent\textbf{Main results. }We describe a new algorithm \textsf{BooleanSolve} that solves Boolean MQ for determined or overdetermined systems ($m=\alpha n$ with $\alpha\ge1$). We show how to adapt it to solve the Boolean MQ SAT problem.
This algorithm has deterministic and Las Vegas variants, depending on the choice of some linear algebra subroutines. Our main result is:
\begin{thm} The Boolean MQ Problem is solved by Algorithm {\sf BooleanSolve}. If~$m=n$ and the system fulfills algebraic assumptions detailed in Theorem~\ref{thm:globalcompl}, then this algorithm uses a number of arithmetic operations in~${\mathbb F}_2$ that is:
\begin{itemize}
\item $O(2^{0.841n})$ using the deterministic variant;
\item of expectation $O( 2^{0.792 n})$ using the \emph{Las Vegas} probabilistic variant.
\end{itemize}
\end{thm}
Recall that for a probabilistic algorithm of the Las Vegas type, the result is always correct, but the complexity is a random variable. Here its expectation is controlled well. 

\noindent\textbf{Outline.}
Our algorithm is a variant of the hybrid
approach by \cite{BetFauPer09,BetFauPer12}: we specialize the last \(k\) variables to all possible values, and check the consistency of the specialized overdetermined systems~$(\tilde{f}_1,\dots,\tilde{f}_m)$ in the remaining variables~$x_1,\dots,x_\ell$.

This consistency check is done by searching for polynomials~$h_1,\dots,h_{m+\ell}$ in $x_1,\dots,x_\ell$ such that
\begin{equation}\label{consistency}
h_1\tilde{f}_1+\dots+h_m\tilde{f}_m+h_{m+1}x_1(1-x_1)+\dots+h_{m+\ell} x_\ell(1-x_\ell)=1.
\end{equation}
If such polynomials exist then obviously the system is not consistent. Given a bound~$d$ on the degrees of the polynomials~$h_i\tilde{f}_i$ and $h_{m+i}x_i(1-x_i)$, the existence of the $h_i$ can be checked by linear algebra. The corresponding matrix is known as the Macaulay matrix in degree \(d\). It 
is a matrix whose rows contain the coefficients of the
polynomials \(\tilde{f}_{i}\) and $x_i(1-x_i)$ multiplied by all monomials of degree
at most $d-2$, each column corresponding to a monomial of degree at most~$d$. Taking into account the special shape of the polynomials~$x_i(1-x_i)$ leads to a more compact variant that we call the boolean Macaulay matrix (see Section \ref{sec:algo}).

When linear algebra on the Macaulay matrix in degree~$d$ produces a solution of~\eqref{consistency}, the corresponding~$h_i$'s give a certificate of inconsistency.
Otherwise, our algorithm proceeds with an exhaustive search in the remaining variables.
In summary, our algorithm is a partial exhaustive search where the Macaulay matrices
permit to prune branches of the search tree. The correctness of the algorithm is clear.

The key point making the algorithm efficient is the choice of~\(k\) and $d$. If $d$ is large, then the cost of the linear algebra stage becomes high. If $d$ is small, the matrices are small, but many branches with no solutions are not pruned and require an exhaustive search. This is where we use the relation between the Macaulay matrix and Gr\"obner bases. We define a \emph{witness degree}
$\dwit$, which has the property that any polynomial in a minimal
Gr\"obner basis of the system 
is obtained as a linear
combination of the rows of the Macaulay matrix in degree \(\dwit \).
Hilbert's Nullstellensatz states that the system has no solution \emph{if and only if}
$1$ belongs to the ideal generated by the polynomials, which implies that~1 is a linear combination of the rows of the Macaulay matrix in degree \(\dwit\), making $\dwit$ an upper bound for the choice of~$d$ in~\eqref{consistency}.

Our complexity estimates rely on a good control of the witness degree. For a homogeneous polynomial ideal, the classical Hilbert function of the degree $d$ is the dimension of the vector space obtained as the quotient of the polynomials of degree~$d$ by the polynomials of degree~$d$ in the ideal.  
The witness degree is bounded by the first degree where the Hilbert function of the ideal generated by the homogenized equations is~0.
Under the algebraic assumption of boolean semi-regularity (see Definition \ref{def:bool semi
reg}), we obtain an explicit expression for the generating series of the Hilbert function, known as the Hilbert series of the ideal. From there, in Proposition \ref{prop:dregh}, using the saddle-point method as in~\citep{BarFauSal04,BarFauSal05,Bar04}, we
show 
 that when  
   \(m= \alpha n\) and \(n\rightarrow\infty\), the witness degree behaves like \({\dwit}\le c_{\alpha}n\) for a constant \(c_\alpha\) that we determine explicitly.
Informally, boolean semi-regularity amounts to demanding a ``sufficient'' independence of the equations. In the case of infinite fields, a classical conjecture by \cite{Fro85} states that generic systems are semi-regular. In our context where the field is~$\mathbb{F}_2$, we give strong experimental evidence (Section \ref{sec:numBeta}) that for $n$ sufficiently large, boolean semi-regularity holds with probability very close to~1 for random systems.  
Thus, our complexity estimates for boolean semi-regular systems apply to a large class of systems in practice.

Once the witness degree is controlled, the size of the Macaulay matrix depends only on the choice of~$k$ and the optimal choice  depends on the complexity of the linear algebra stage.
In the Las Vegas version of Algorithm \textsf{BooleanSolve}, we exploit the sparsity of this matrix by using a variant of Wiedemann's algorithm~\citep{GieLobSau98} (following \cite{Wie86,KalSau91,Vil97}) for solving singular linear systems. 
In the deterministic version, we do not know of efficient ways to take advantage of the sparsity of the matrix, whence a slightly higher complexity bound.
We can then draw conclusions and obtain  a complexity estimate of the algorithm depending on $k/n$ and \(n\) (Proposition \ref{prop:complFirstStep}). The optimal value for \(k\) is \(\simeq 0.45\,n \) in the Las Vegas setting and \(\simeq 0.59\,n \) in the deterministic variant, completing the proof of our main theorem.

The complexity analysis is especially
important for practical applications in multivariate Cryptology based
on the Boolean MQ problem, since it shows that in order to reach a security of
$2^{s}$ (with $s$ large), one has to construct systems of boolean
quadratic equations with at least $s/0.7911\simeq 1.264 s$ variables.

\medskip

\noindent\textbf{Related works.}
Due to its practical importance, many algorithms have been designed to
solve the MQ problem in a wide range of contexts. First, generic
techniques for solving polynomial systems can be used. In particular,
Gr\"obner basis algorithms (such as Buchberger's algorithm
\citep{Buc65}, $F_4$ \citep{Fau99}, $F_5$ \citep{Fau02}, and FGLM
\citep{FauGiaLazMor93}) are well suited for this task. For instance,
the $F_5$ algorithm has broken several challenges of the HFE 
public-key cryptosystem \citep{FauJou03}. In the cryptanalysis
context, the XL algorithm~\citep{KipSha99} (which can be seen as a variant of Gr\"obner
basis algorithms \citep{ArsFauImaKawSug04}) has given rise to a large family of variants. All these
techniques are closely related to the Macaulay matrix,
introduced by \cite{Mac03} as a tool for elimination.
In order to reduce the cost of linear algebra for the efficient computation of the
resultant of multivariate polynomial systems, the idea of using
Wiedemann's algorithm on the Macaulay matrix has been proposed by \cite{CKY89}; however since the specificities of the Boolean
case are not taken into account, the complexity of applying \citep{CKY89}
to quadratic equations is \(O(2^{4\,n})\).

\cite{YanChe04} propose a heuristic analysis of the
FXL algorithm leading them to an upper bound $O(2^{0.875n})$ for the
complexity of solving the MQ problem over $\mathbb{F}_2$. In
particular, they give an explicit formula for the Hilbert series of
the ideal generated by the polynomials. However, 
the exact assumptions that
have to be verified by the input systems are unclear. Also, similar
results have been announced in \cite[Section 2.2]{CouYanChe04}, but
the analysis there relies on algorithmic assumptions (e.g., row
echelon form of sparse matrices in quadratic complexity) that are not
known to hold currently. Under these assumptions, the authors show that the
best trade-off between exhaustive search and row echelon form
computations in the FXL algorithm is obtained by specializing $0.45 n$ variables. This is the same value we obtain and prove with our algorithm. Also, a limiting behavior of the cost of the hybrid approach is obtained in~\cite{BetFauPer12} when the size of the finite field is big enough; these results are not applicable
over \(\mathbb{F}_{2}\).

Other algorithms have been proposed when the system has additional structural
properties. In particular, the Boolean MQ
problem also arises in satisfiability problems, since boolean
quadratic polynomials can be used for representing constraints. 
In these contexts, the systems are sparse and for such systems of higher degree the $2^n$ barrier has been broken~\citep{Sem08,Sem09}; similar results also exist for the $k$-SAT problem. Our algorithm does not exploit the extra structure induced by this type of sparsity and thus does not improve upon those results.

\medskip
\noindent\textbf{Organization of the article.} 
The main algorithm and the algebraic tools that are used throughout
the article are described in Section~\ref{sec:algo}. Then a complexity
analysis is performed in Section~\ref{sec:compl} by studying the
asymptotic behaviour of the witness degree and the sizes of the
Macaulay matrices involved, under algebraic assumptions. In
Section~\ref{sec:expmodel}, we provide a conjecture and strong
experimental evidence that these algebraic assumptions are verified
with probability close to $1$ for $n$ sufficiently large. Finally, in
Section \ref{sec:appli} we propose an extension of the main algorithm
that improves the quality of the linear filtering when $n$ is small. We
also show how the complexity results from Section \ref{sec:compl} can
be applied to the cryptosystem QUAD, leading to an evaluation of the sizes of the
parameters needed to reach a given level of security.

%% file: notations.tex
Let $m$ and $n$ be two positive integers and let
$R$ be the ring $\mathbb{F}_2[x_1,\ldots, x_n]$.  In the
following, the notation $\mathsf{Monomials}(d)$ stands for the set of monomials in
$R$ of degree at most $d$. 

Since we are looking for solutions of the system in $\mathbb{F}_2$
(and not in its algebraic closure), we have to take into account the relations
$x_i^2-x_i=0$. Therefore, we consider the application $\varphi$ 
mapping a monomial to its square-free part ($\varphi(\prod_{i=1}^n x_i^{a_i})=\prod_{i=1}^n x_i^{\min(a_i,1)}$) and extended to $R$ by
linearity.

If $(f_1,\ldots, f_m)\in\mathbb{F}_2[x_1,\ldots, x_n]^m$ is a system of polynomials, its homogenization is denoted by the system $(f_1^{(h)},\ldots,
f_m^{(h)})\in \mathbb{F}_2[x_1,\ldots, x_n,h]$ and is
defined by
$$f_i^{(h)}(x_1,\ldots, x_n, h)=h^{\deg{(f_i)}} f_i\left(\frac{x_1}{h},\ldots, \frac{x_n}{h}\right).$$

In the sequel, we consider the classical \emph{grevlex} monomial ordering (\textbf{g}raded \textbf{rev}erse \textbf{lex}icographical), as defined for instance in \cite[\S2.2, Defn. 6]{CoxLitShe97}.
Also, if $f$ is a polynomial, $\LM(f)$ denotes its leading monomial for that order. If $I$  is an ideal, then $\LM(I)$ denotes the ideal generated by the leading monomials of all polynomials in $I$.

%% file: macaulay.tex
\begin{defn}
  Let $(f_1,\ldots, f_m)$ be 
  polynomials in $R$. The \emph{boolean Macaulay
  matrix} in degree~$d$ (denoted by $\mathsf{Macaulay}(d)$) is the matrix
  whose rows contain the coefficients of the polynomials $\{\varphi(t
  f_j)\}$ where $1\leq j\leq m$, $t$ is a squarefree monomial, and $\deg(t f_j)=d$. The columns
 correspond to the
  squarefree monomials in $R$ of degree at most $d$ and are ordered in descending order with respect to the grevlex ordering. The element in the row corresponding to $\varphi(t f_j)$ and the
  column corresponding to the monomial $m$ is the coefficient of $m$
  in the polynomial $\varphi(t f_j)$.
\end{defn}

Note that the boolean Macaulay matrix can be obtained as a submatrix
of the classical Macaulay matrix of the system $\langle
f_1,\ldots,f_m, x_1^2-x_1,\ldots, x_n^2-x_n\rangle$ after Gaussian reduction by
the rows corresponding to the polynomials $(x_1^2-x_1,\ldots,
x_n^2-x_n)$.

\begin{lem}\label{lem:inconsistent}
  Let $\mathsf M$ be the $r_{\mathsf{Mac}}\times c_{\mathsf{Mac}}$ boolean Macaulay matrix of the system $(f_1,\ldots,
  f_m)$ in degree $d$. 
  Let $\mathbf r$
  denote the $1\times c_{\mathsf{Mac}}$ vector 
$\mathbf r=(0,\dots,0,1)$.
  If the linear system $\mathbf u\cdot\mathsf M=\mathbf r$ has a solution,
  then the system $f_1=\dots = f_m=0$ has no solution in $\mathbb{F}_2^n$.
\end{lem}

\begin{proof}
If the system $\mathbf u\cdot\mathsf M=\mathbf r$ has a solution, then there exists a linear combination of the rows of the Macaulay matrix which yields the constant polynomial $1$. Therefore, $1\in\langle f_1,\ldots, f_m,x_1^2-x_1,\ldots, x_n^2-x_n\rangle$.
\end{proof}
\subsection{Witness degree}
We consider an indicator of the 
complexity of affine polynomial systems: the \emph{witness degree}. It
has the property that a Gr\"obner basis of the ideal generated by the
polynomials can be obtained by performing linear algebra on the
Macaulay matrix in this degree. In particular, if the system has no
solution, then the witness degree is closely related to the classical
\emph{effective
  Nullstellensatz} (see e.g., \cite{Jelonek2005}).

\begin{defn}\label{defn:dlin}
Let $\mathbf F=(f_1,\ldots, f_m, x_1^2-x_1,\ldots, x_n^2-x_n)$ be a sequence of polynomials and $I=\langle \mathbf F\rangle$ the ideal it generates.
Denote by $I_{\leq d}$ and by $J_{\leq d}$ the $\mathbb{F}_2$-vector spaces defined by
$$\begin{array}{rccl}
I_{\leq d}&=&\{p \mid& p\in I, \deg(p)\leq d\},\\
J_{\leq d}&=&\{p \mid& \exists h_1,\ldots, h_{m+n}, \forall i\in \{1,\ldots, m+n\}, \deg(h_i)\leq d-2,\\
&&&p=\sum_{i=1}^m h_i f_i + \sum_{j=1}^n h_{m+j} (x_j^2-x_j)\}.\end{array}$$

\noindent We call \emph{witness degree} $(\dwit)$ of $\mathbf F$ the smallest integer $d_0$ such that $I_{\leq d_0}=J_{\leq d_0}$ and $\langle\{\LM(f) \mid f\in I_{\leq d_0}\}\rangle=\LM(I)$.
\end{defn}

Consider a row echelon form of the boolean Macaulay matrix in degree $d$ of
the system $(f_1,\dots,f_m)$ of polynomials. Then the first nonzero element of each row corresponds to a leading monomial of an element of~$I$, belonging to~$\LM(I)$. For large enough $d$,
Dickson's lemma \cite[\S2.4, Thm. 5]{CoxLitShe97} implies that the collection of those monomials up to degree~$d$ generates~$\LM(I)$ and thus the polynomials corresponding to those rows together with
$\{x_1^2-x_1,\ldots, x_n^2-x_n\}$ form a Gr\"obner basis of~$I$
with
respect to the grevlex ordering. Another interpretation of the \emph{witness degree} is that it is precisely the smallest such $d$.

%% file: algorithm.tex
\begin{algorithm}
\caption{BooleanSolve}
\label{algo:booleansolve}
\begin{algorithmic}[1]
\Require $m, n, k\in\mathbb N$ such that $m\geq n>k$ and $f_1,\ldots,f_m$ quadratic polynomials in $\mathbb{F}_2[x_1,\ldots, x_n]$.
\Ensure The set of boolean solutions of the system $f_1=\dots=f_m=0$.
\State $S:=\emptyset$. 
\State ${d_{0}}:=$ index of the first nonpositive coefficient in the series expansion at~0 of the rational function $\frac{(1+t)^{n-k}}{(1-t)(1+t^2)^{m}}$.
\ForAll{$(a_{n-k+1},\ldots, a_n)\in\mathbb F_2^{k}$}
\For{$i$ from $1$ to $m$}
\State $\tilde{f}_i(x_1,\ldots,x_{n-k}):=f_i(x_1,\ldots, x_{n-k},a_{n-k+1},\ldots, a_n)\in \mathbb{F}_2[x_1,\ldots, x_{n-k}]$.
\EndFor 
\State $\mathsf{M}:=$ boolean Macaulay matrix of $(\tilde{f}_1,\ldots,\tilde{f}_m)$ in degree ${d_{0}}$.
\If{the system $\mathbf u\cdot\mathsf M=\mathbf r$ is inconsistent}\Comment{$\mathbf r$ as defined in Lemma~\ref{lem:inconsistent}}
\State $T:=$ solutions of the system ($\tilde{f}_1=\cdots=\tilde{f}_m=0$) (exhaustive search).
\ForAll{$(t_1,\ldots, t_{n-k})\in T$}
\State $S:=S\cup \{(t_1,\ldots, t_{n-k},a_{n-k+1},\ldots, a_n)\}$.
\EndFor
\EndIf
\EndFor
\State \Return $S$.
\end{algorithmic}
\end{algorithm}

Our algorithm is given in Algorithm \ref{algo:booleansolve}.
The general principle is to perform an exhaustive search in two steps,
using a test of consistency of the Macaulay matrix to prune most of the
branches of the second step of the search.

When the system $\mathbf u\cdot\mathsf M=\mathbf r$ is
consistent, the corresponding branch of the searching tree is not explored.
In that case, by Lemma~\ref{lem:inconsistent}, any solution of the linear system $\mathbf u\cdot \mathsf
M=\mathbf r$ can be used as a certificate that there exists no solution of the
polynomial system $f_1=\dots=f_m=0$ in this branch.

\begin{prop}\label{algo_correct}Algorithm~{\sf{BooleanSolve}} is correct and solves the Boolean MQ\ problem.
\end{prop}
\begin{proof}
By Lemma~\ref{lem:inconsistent}, if the test in line~8 finds the linear system to be consistent, then there can be no solution with the given values of~$(a_{n-k+1},\dots,a_n)$. Otherwise, the exhaustive search proceeds and cannot miss a solution. It is important to note that the choice of the actual value~${d}_0$ does not have any impact on the correctness of the algorithm.
\end{proof}

Algorithm~{\sf{BooleanSolve}} is easily be adapted to solve the Boolean MQ\  SAT problem by replacing
lines 9-12 of the previous algorithm by:
\begin{algorithm}[H]
\begin{algorithmic}[1]
\setcounter{ALG@line}{8}
\State $T:=$ at least one solution of the system ($\tilde{f}_1=\cdots=\tilde{f}_m=0$) (modified exhaustive search).
\If{\(T\neq\emptyset\)} 

\State  \Return \(\{(t_1,\ldots, t_{n-k},a_{n-k+1},\ldots, a_n)\mid
(t_1,\ldots, t_{n-k})\in T
\}\)
\EndIf
\end{algorithmic}
\end{algorithm}

%% file: wiedemann.tex
\label{sec:wiedemann}
The choice of the algorithm to test whether the sparse linear system $\mathbf u\cdot\mathsf M=\mathbf r$
is consistent or not is crucial for the efficiency of Algorithm {\sf BooleanSolve}. 
A simple deterministic algorithm consists in computing a row echelon form of the matrix: the linear system is consistent if and only if the last nonzero row of the row echelon form is equal to the vector $\mathbf r$.
We show in Section~\ref{sec:compl} that this is sufficient to pass below the $2^n$ complexity barrier. 
We recall for future use the complexity of this method.
\begin{prop}[\cite{Storjohann2000}, Proposition 2.11]\label{prop:row_ech} The row echelon form of an $N\times M$ matrix over a field~$k$ can be computed in~$O(NMr^{\theta-2})$ arithmetic operations in~$k$, where~$r$ is the rank of the matrix and $\theta\le3$ is such that any two~$n\times n$ matrices over~$k$ can be multiplied in~$O(n^\theta)$ arithmetic operations in~$k$.
\end{prop}
Here, $\theta=3$ is the cost of classical matrix multiplication and in this case a simple Gaussian reduction to row echelon form is sufficient. The best known value for $\theta$ has been~2.376 for a long time, by a result of~\cite{CoWi90}. Recent improvements of that method by~\cite{Stothers10,VassilevskaWilliams11} have decreased it to~2.3727, but this does not have a significant impact on our analysis.

This result does not exploit the sparsity of Macaulay matrices. We do not know of an efficient deterministic algorithm for row reduction that exploits this sparsity. 
Instead, we use an efficient Las Vegas variant of Wiedemann's algorithm due to \cite{GieLobSau98}, whose 
specification is summarized in Algorithm~{\sf TestConsistency}. In this algorithm, the matrix~$A$ is given by two black boxes performing the operations~$x\mapsto Ax$ and $u\mapsto A^tu$. The complexity is expressed in terms of the number of evaluations of these black boxes, which in our context will each have a cost bounded by the number of nonzero coefficients of Macaulay matrices. The algorithm is presented in \citep{GieLobSau98} for matrices with entries in an arbitrary field. We specialize it here in the case where the field is $\mathbb{F}_2$. The key ideas are a preconditioning of the matrix by multiplying it by random Toeplitz matrices and working in a suitable field extension to get access to sufficiently many points for picking random elements.

\begin{algorithm}
\caption{TestConsistency \citep{GieLobSau98}}
\label{algo:testinconsistency}
\label{alg:sparsesolver}
\begin{algorithmic}[1]
\Require\begin{itemize}
\item A black box for $\mathbf x\mapsto \mathsf A\cdot \mathbf x$, where $\mathsf A\in \mathbb K^{N\times N}$.
\item A black box for $\mathbf u\mapsto \mathsf A^t\cdot\mathbf u$.
\item $\mathbf b\in \mathbb K^{N\times 1}$.
\end{itemize}
\Ensure \begin{itemize}
\item (``consistent'',$\mathbf x$) with $\mathsf A\cdot \mathbf x=\mathbf b$ if the system has a solution
\item (``inconsistent'',$\mathbf u$) if the system does not have a solution, with $\mathbf u^t\cdot\mathsf A=0$ and $\mathbf u^t\cdot\mathbf b\neq 0$, certifying the inconsistency.
\end{itemize}
\end{algorithmic}
\end{algorithm}

\begin{prop}[\cite{GieLobSau98}]\label{thm:wiedemann}
Algorithm \ref{alg:sparsesolver} 
determines the consistency of an $N\times N$ matrix with expected
complexity $O(N \log N )$ evaluations of the black boxes and 
$O(N^2\log^2N\log\log N)$ additional operations in
$\mathbb{F}_2$. 
\end{prop}

Macaulay matrices are rectangular. We therefore first make them square by padding with zeroes. The complexity estimate is then used with~$N$ the maximum of the row and column dimensions of the matrices.

%% file: sizemacaulay.tex
\begin{prop}\label{prop:sizemac}
Let $(f_1,\dots,f_m)$ be quadratic polynomials in~${\mathbb F}_2[x_1,\dots,x_n]$. 
Denote by $r_{\mathsf{Mac}}$ (resp. $c_{\mathsf{Mac}}$, $s_{\mathsf{Mac}}$)
the number of rows (resp. columns, number of nonzero entries) of the associated boolean Macaulay matrix in degree~$d$. If~$1\le d< n/2$, then 
\begin{equation}\label{eq:mac}
c_{\mathsf{Mac}}<\frac{1-x}{1-2x}\binom{n}{d},\,\,
r_{\mathsf{Mac}}<{m}\frac{x^2}{(1-2x)(1-x)}\binom{n}{d},\,\, 
s_{\mathsf{Mac}}<{mn^2}\frac{x^2}{(1-2x)(1-x)}\binom{n}{d},
\end{equation}
where $x=d/n$.
\end{prop}
\begin{proof}
The number of columns of the boolean Macaulay matrix is simply the number of squarefree monomials of degree at most~$d$ in $n$ variables. The number of rows is that same number of monomials for degree~$d-2$, multiplied by the number~$m$ of polynomials. Finally, each row corresponding to a polynomial~$f_i$ has a number of nonzero entries bounded by the number of squarefree monomials of degree at most~2 in $n$ variables. 
Standard combinatorial counting thus gives
\begin{equation}\label{eq:size_matrix}
c_{\mathsf{Mac}}=\sum_{i=0}^d{\binom{n}{i}},\qquad r_{\mathsf{Mac}}=m\sum_{i=0}^{d-2}{\binom{n}{i}},\qquad
s_{\mathsf{Mac}}\le \left(1+n+\binom{n}{2}\right)r_{\mathsf{Mac}}\le n^2r_{\mathsf{Mac}},
\end{equation}
where in the last inequality we use the fact that~$n\ge 2$. 
Now, the bounds come from a well-known inequality on binomial coefficients: for $0\le d< n/2$,
\[\sum_{i=0}^{d}{\binom{n}{i}}<\frac{1}{1-d/(n-d)}\binom{n}{d}.\]
Indeed, the sequence $\binom{n}{i}$ is increasing for $0\le i\le n/2$. Factoring out $\binom{n}{d}$ leaves a sum that is bounded by the geometric series $1+d/(n-d)+\dotsb$. This gives the bound for~$c_{\mathsf{Mac}}$. The bound for~$r_{\mathsf{Mac}}$ is obtained by evaluating this bound at~$d-2$, writing $\binom{n}{d-2}$ as a rational function times $\binom{n}{d}$ and finally bounding $x(x-1/n)/((1-2x+4/n)((1-x)+1/n))$ by $x^2/((1-2x)(1-x))$.
\end{proof}

%% file: linearity-new2.tex
First, we prove that the witness degree can be upper bounded by the
so-called \emph{degree of regularity} of the homogenized system.  Here
and subsequently, we call \emph{dimension} of an ideal $I\subset R$
the Krull dimension of the quotient ring $R/I$ (see e.g., \cite[\S8]{Eis95}).

\begin{defn}
  The \emph{degree of regularity} $\dreg(I)$ of a
  homogeneous ideal $I$ of dimension $0$ is defined as the smallest integer $d$ such that all homogeneous polynomials of degree $d$ are in $I$.
\end{defn}

\begin{prop}\label{prop:dlindregh}
  Let $\mathbf F=\left(f_1,\ldots, f_m, x_1^2-x_1,\ldots,
    x_n^2-x_n\right)$ be a sequence of polynomials such that the
  system $\mathbf F=0$ has no solution. Then the ideal generated by
  the homogenized
  system   $$I^{(h)}=\left\langle f_1^{(h)},\ldots, f_m^{(h)},x_1^2-x_1
    h,\ldots, x_n^2-x_n h\right\rangle$$
has dimension $0$ and $\dwit(\mathbf F)\leq \dreg(I^{(h)})$.
\end{prop}

\begin{proof}
  By Hilbert's  Nullstellensatz, the ideal $I$ generated by
  $\mathbf F$ contains $1$ (hence $1$ is a Gr\"obner basis of
  $I$). Therefore, there exists $\alpha\in\mathbb N\setminus\{0\}$ such that
  $h^\alpha\in I^{(h)}$. Consequently, for the grevlex ordering,
  $\langle x_1^2,\ldots,
  x_n^2,h^\alpha\rangle\subset\LM(I^{(h)})$ and thus the
  dimension of $\LM(I^{(h)})$ is $0$. As a consequence (see \cite[\S9.3, Prop.~9]{CoxLitShe97}),
  $\dim(I^{(h)})=\dim(\LM(I^{(h)}))=0$.

  Let $G^{(h)}$ be a minimal Gr\"obner basis of the homogenized ideal
  $I^{(h)}$ for the grevlex ordering. By definition of the degree of
  regularity, there
  exist polynomials $\ell_i$ and $\ell'_j$ such that
  $h^{\dreg(I^{(h)})}=\sum_{1\leq i\leq m} f_i^{(h)} \ell_i +\sum_{1\leq j\leq n}
  (x_j^2-x_j h)\ell'_j$. The ideal~$I^{(h)}$ being homogeneous, it is possible to find such a combination with $\deg(f_i^{(h)} \ell_i)\leq
  \dreg(I^{(h)}),\deg((x_j^2-x_j h) \ell'_j)\leq \dreg(I^{(h)})$ for all $ i,j$.
Evaluating this identity at~$h=1$ shows that $1$ belongs to 
the vector
  space generated by the rows of the boolean Macaulay matrix in degree
  $\dreg(I^{(h)})$.
\end{proof}

Next, the degree of regularity can be obtained from the classical Hilbert series.
\begin{defn}
Let $R^{(h)}$ be the ring $\mathbb{F}_2[x_1,\ldots, x_n,h]$, and let $R^{(h)}_d$ be the vector space of homogeneous polynomials of degree $d$. Also, for $I\subset R^{(h)}$ a homogeneous ideal, let $I_d$ denote the vector space defined by $I_d=R^{(h)}_d\cap I$.
The Hilbert function $\mathsf{HF}_I$ and the Hilbert series $\mathsf{HS}_I$ of $I$ are defined by
$$\mathsf{HF}_I(d)=\dim (R^{(h)}_d/I_d),\hspace{1.6cm}\mathsf{HS}_I(t)=\sum_{i=0}^\infty \mathsf{HF}_I(d) t^d.$$
\end{defn}
In view of the definition of the degree of regularity, 
if $I$ is a zero-dimensional ideal of $R^{(h)}$, then $\mathsf{HS}_I(t)$ is a polynomial of degree $\dreg(I)-1$.

The next step is to obtain information on the Hilbert series for a
large class of systems. To this end, we consider the so-called
\emph{syzygy module}, which describes the algebraic relations between
the polynomials of a system.

\begin{defn}
  Let $(g_1,\ldots,g_\ell)\in (R^{(h)})^\ell$ be a polynomial
  system. A \emph{syzygy} of $(g_1,\ldots,g_\ell)$ is a $\ell$-tuple
  $(s_1,\ldots, s_\ell)\in (R^{(h)})^\ell$ such that $\sum_{i=1}^\ell
  s_i g_i =0.$ The set of all syzygies of $(g_1,\ldots,g_\ell)$ is a
  submodule of $(R^{(h)})^\ell$. The degree of a syzygy $\mathbf
  s=(s_1,\ldots,s_\ell)$ is defined as $\deg(\mathbf s)=\max_{1\leq i\leq \ell} \deg(g_i s_i)$.
\end{defn}

Obviously, for any such polynomial system, commutativity induces syzygies of the type
\begin{equation}\label{syz-1}
g_i g_j-g_j g_i=0.
\end{equation}

Moreover, for any constant~$a\in\mathbb F_2$ we have the relation $a^2=a$, thus expanding the square of a  polynomial~$\sum_{\alpha\in\mathbb N^k} a_\alpha \mathbf{x}^\alpha\in \mathbb F_2[x_1,\ldots, x_k]$ gives
$\sum_{\alpha\in\mathbb N^k} a_\alpha \mathbf{x}^{2\alpha}$. As a consequence, for a homogeneous quadratic polynomial $f_i^{(h)}=\sum_{1\leq j,k\leq n} a_{j,k}\,x_j x_k+\sum_{1\leq j \leq n} b_j\,x_j h + c\, h^2\in \mathbb F_2[x_1,\ldots, x_n,h]$, we obtain the following syzygy of the system $(f_i^{(h)},x_1^2-x_1h,\dots,x_n^2-x_nh)$:
\begin{equation}\label{syz-2}
(f_i^{(h)}-h^2)f_i^{(h)}+\sum_{1\leq j,k\leq n} a_{j,k} \left(x_k^2(x_j^2-x_j h) +  x_j h(x_k^2-x_k h)\right)+ \sum_{1\leq j \leq n}  b_j h^2(x_j^2-x_j h) = 0.
\end{equation}

\begin{defn}Let $\mathbf{F}^{(h)}=(f_1^{(h)},\dots,f_n^{(h)},x_1^2-x_1h,\dots,x_n^2-x_nh)$ be a system of homogeneous quadratic polynomials over~$\mathbb{F}_2$. We call \emph{trivial syzygies} of $\mathbf{F}^{(h)}$ and note $Syz_{\text{triv}}$ the module generated by the syzygies of types~(\ref{syz-1}) and~(\ref{syz-2}).
\end{defn}

\begin{defn}\label{def:bool semi reg}
  A boolean homogeneous system $(f_1^{(h)},\ldots,f_m^{(h)})$ is called 
\begin{itemize}
\item \emph{boolean semi-regular} in degree $D$ if any syzygy whose degree is less than $D$ belongs to $Syz_{\text{triv}}$;
\item \emph{boolean semi-regular} if it is boolean semi-regular in degree $\dreg(\langle f_1^{(h)},\ldots,f_m^{(h)}, x_1^2- x_1 h,\ldots, x_n^2-x_n h\rangle)$.
\end{itemize}
\end{defn}
\noindent(This notion is slightly different from the \emph{semi-regularity over $\mathbb{F}_2$} defined in \citep{BarFauSal04,BarFauSal05}.)

\smallskip

In the sequel we use the following notations: if $S\in \mathbb Z[[t]]$ is a power series, then $[S]$ denotes the series obtained by truncating $S$ just before the index of its first nonpositive coefficient. Also, $[t^d]S(t)$ denotes the coefficient of $t^d$ in $S$.

\begin{prop}\label{prop:HSboolSemireg}
Let $(f_1^{(h)},\ldots, f_m^{(h)})$ be a boolean homogeneous system. Let $D_0$ denote the degree of regularity of the system $(f_1^{(h)},\ldots, f_m^{(h)}, x_1^2-x_1 h,\ldots, x_n^2-x_n h)$.
If the systems $(f_1^{(h)},\ldots, f_{i-1}^{(h)},f_i^{(h)}-h^2)$ and $(f_1^{(h)},\ldots, f_{i-1}^{(h)},f_i^{(h)})$ are  $D_0-2$ (resp.  $D_0$)-boolean semi-regular for each $i\in \{2,\ldots, m\}$,
then the Hilbert series of the homogeneous ideal $\langle f_1^{(h)},\ldots, f_m^{(h)},x_1^2-x_1 h,\ldots, x_n^2-x_n h\rangle$ is
$$\mathsf{HS}_{n,m}(t):=\left[\frac{(1+t)^{n}}{(1-t)(1+t^2)^{m}}\right].$$
\end{prop}

\begin{proof}
Let $S_i$ (resp. $S_i'$) denote the system $(f_1^{(h)},\ldots,f_i^{(h)}, x_1^2-x_1 h,\ldots, x_n^2- x_n h)$ (resp. $(f_1^{(h)},\ldots,f_i^{(h)}-h^2, x_1^2-x_1 h,\ldots, x_n^2- x_n h)$).  
The general framework of this proof is rather classical: we prove by induction on $i$ and $d$ that for all $i\le m$ and $d<D_0$, $\mathsf{HF}_{\langle S_i\rangle}(d)=\mathsf{HF}_{\langle S_i'\rangle}(d)=[t^d]\frac{(1+t)^{n}}{(1-t)(1+t^2)^{i}}$. 

First, notice that a basis of the $\mathbb F_2$-vector space $R/\langle x_1^2-x_1 h,\ldots, x_n^2-x_n h\rangle$ is the set of monomials $\mathfrak S=\{x_1^{\delta_1}\dotsm x_n^{\delta_n}h^\ell \mid \delta_1,\ldots, \delta_n\in \{0,1\}, \ell\in \mathbb N\}$. The generating function of this set is 
$$\sum_{\mathfrak m\in \mathfrak S} t^{\deg(\mathfrak m)} = \frac{(1+t)^{n}}{(1-t)}.$$
Therefore, the initialization of the recurrence comes from the relations 
$$\begin{cases}\mathsf{HF}_{\langle x_1^2-x_1 h,\ldots, x_n^2-x_n h\rangle}(d)=[t^d]\frac{(1+t)^{n}}{(1-t)}\text{ for all }d\in \mathbb N;\\\mathsf{HF}_{\langle S_i\rangle}(0)=\mathsf{HF}_{\langle S_i'\rangle}(0)=1\text{ and }\mathsf{HF}_{\langle S_i\rangle}(1)=\mathsf{HF}_{\langle S_i'\rangle}(1)=n+1\text{ for all }i\le m.\end{cases}$$

In the following, $2\le d<D_0$ and $1\le i\le m$ are two integers, and
we assume by induction that for all $(\ell, j)\in \mathbb N^2$ such
that 
$\ell < d$ or ($\ell=d$ and $j<i$), we have
$$\mathsf{HF}_{\langle S_j\rangle}(\ell)=\mathsf{HF}_{\langle S_j'\rangle}(\ell)=[t^\ell]\frac{(1+t)^{n}}{(1-t)(1+t^2)^j}.$$

Consider the following sequences
  \small$$\begin{array}{clr}0\rightarrow  R^{(h)}_{d-2}/(
  S_{i-1}+\langle 
  f_i^{(h)}-h^2\rangle)_{d-2}\xrightarrow{\times 
    f_i^{(h)}}& R^{(h)}_d/( S_{i-1})_d\rightarrow
   R^{(h)}_d/( S_i)_d\rightarrow 0\\
0\rightarrow  R^{(h)}_{d-2}/(
  S_{i-1}+\langle 
  f_i^{(h)}\rangle)_{d-2}\xrightarrow{\times 
    (f_i^{(h)}-h^2)}& R^{(h)}_{d}/( S_{i-1})_{d}\rightarrow
   R^{(h)}_{d}/( S_i')_{d}\rightarrow 0,& 
\end{array}
$$ \normalsize
where the last arrow of each sequence is the canonical projection.
Let $g$ be in the kernel of the application
$$ R^{(h)}_{d-2}/(
S_{i-1}+\langle f_i^{(h)}-h^2\rangle)_{d-2}\xrightarrow{\times
  f_i^{(h)}} R^{(h)}_d/( S_{i-1})_d.$$ Then $g f_i^{(h)}$ belongs to
$(S_{i-1})_d$, which implies that there exist polynomials $g_1,\dots, g_{i-1},$ $h_1,\dots, h_n$ such
that $(g_1,\ldots, g_{i-1},g,h_1,\ldots, h_n)$ is a syzygy of degree $d$ of the system
$S_i$. By the boolean semi-regularity assumption, this syzygy belongs
to $Syz_{\text{triv}}$, and hence $g\in \langle S_{i-1}\rangle+\langle
f_i^{(h)}-h^2\rangle$. Therefore the application $\times f_i^{(h)}$ is
injective and the first sequence is exact. One can prove similarly that the second sequence
is also exact.

  These exact sequences yield relations between the Hilbert functions:
\begin{align}\label{eq:hilb2}
\mathsf{HF}_{ S_{i}'}(d-2)-\mathsf{HF}_{ S_{i-1}}(d)+\mathsf{HF}_{ S_i}(d)&=0,\\
\label{eq:hilb1}
\mathsf{HF}_{ S_{i}}(d-2)-\mathsf{HF}_{ S_{i-1}}(d)+\mathsf{HF}_{ S_i'}(d)&=0.
\end{align}
Moreover, we have the relation
\begin{equation}
  \label{eq:genser}
 [t^{\ell}]\frac{(1+t)^{n}}{(1-t)(1+t^2)^j} = [t^{\ell}]\frac{(1+t)^{n}}{(1-t)(1+t^2)^{j-1}} - [t^{\ell-2}]\frac{(1+t)^{n}}{(1-t)(1+t^2)^j}.
\end{equation}
Using Relations~\eqref{eq:hilb2} and \eqref{eq:hilb1}, and the induction hypothesis, we get the
desired result.

\smallskip

The proof is completed by showing that $D_0$ is equal to the index of the first
nonpositive coefficient of $\mathsf{HF}_{S_m}(t)$. First, by definition of
the degree of regularity, the coefficients $[t^d]\mathsf{HF}_{S_m}(t)$
are zero for $d\ge D_0$. Next, that the coefficient
$[t^{D_0}]\frac{(1+t)^{n}}{(1-t)(1+t^2)^{m}}$ is nonpositive
follows from the following property (easily proved by
induction on $i$, $0\le i \le m$ using~(\ref{eq:hilb1}--\ref{eq:genser})):
$$ [t^{D_0}] \frac{(1+t)^{n}}{(1-t)(1+t^2)^{i}} \le \mathsf{HF}_{ S_i}(D_0). $$
\end{proof}

Putting everything together, we have obtained the following.
\begin{cor}With the same notation as in Proposition~\ref{prop:dlindregh}, if the homogenized system 
  verifies the conditions of Proposition \ref{prop:HSboolSemireg}, then the witness degree of the system~$$(f_1,\dots,f_m,x_1^2-x_1,\dots,x_n^2-x_n)$$ is bounded by the degree of the polynomial~$\mathsf{HS}_{n,m}(t)$. 
\end{cor}

At this stage, it might seem that choosing the degree of~$\mathsf{HS}_{n-k,m}$ for~$d_0$ in Algorithm {\sf BooleanSolve} amounts to making a very strong assumption on the nature of the systems obtained by specialization followed by homogenization. In Section \ref{sec:expmodel}, we discuss experiments showing that this assumption is actually quite reasonable.

\medskip Finally, in order to compute the asymptotic behavior of our complexity estimates in the next section, we need the following.
\begin{prop}\label{prop:dregh}
Let $\alpha\geq 1$ be a real number. 
Then, as $n\rightarrow\infty$,
\[\begin{array}{rl}&\deg\left( \mathsf{HS}_{n,\lceil\alpha n\rceil}(t)\right)\sim M(\alpha)n,\\
\text{with}\quad &\displaystyle M(x):=-x+\frac{1}{2}+\frac{1}{2}\sqrt{2 x^2-10 x-1+ 2(x+2)\sqrt{x(x+2)}}.\end{array}\]
\end{prop}
\begin{proof}
We follow the approach of \cite{BarFauSal04,BarFauSal05}.
We start from a representation of the coefficient as a Cauchy integral:
$$[t^d]\frac{(1+t)^{n}}{(1-t)(1+t^2)^{m}}=\frac{1}{2\pi\imath}\oint \frac{(1+z)^{n}}{(1-z)(1+z^2)^{\lceil \alpha n\rceil}}\frac{1}{z^{d+1}}dz,$$
where the contour is a circle centered in $0$ whose radius is smaller than $1$. 
We are searching for a value of $d$ where this integral vanishes, for large $n$. 
We first estimate the asymptotic behaviour of the integral for fixed $d$. The integrand has the form
$\exp(n f(z))$ with
$$f(z)=\log(1+z)-\frac{\lceil \alpha n\rceil}{n}\log(1+z^2)-\frac{\log(1-z)+(d+1)\log(z)}{n}.$$
As $n$ increases, the integral concentrates in the neighborhood of one or several saddle points, solutions to
the saddle-point equation $zf'=0$, which rewrites
\begin{equation}\label{eq:sp}
\frac{d}{n}=\frac{z}{1+z}-\frac{2 \frac{\lceil \alpha n\rceil}{n} z^2}{1+z^2}-\frac{1-2z}{n(1-z)}=:\phi(z)+O(1/n).
\end{equation}
In \cite{BarFauSal04}, it is shown that for the contributions of saddle points to cancel out, two of them must coalesce and give rise to a double saddle point, given by the smallest positive double real root of the saddle-point equation, which is therefore such that $(z f')'=0$.  When $n$ grows, the
solutions of this equation tend towards the roots of $\phi'(z)=0$.
Let $z_0$ be the smallest positive real root of this equation. The saddle-point equation~\eqref{eq:sp} then gives $d\sim \phi(z_0) n$.
Finally, eliminating~$z_0$ using $\phi'(z_0)=0$ by a resultant computation yields
 $$d\sim \left(-\alpha+\frac{1}{2}+\frac{1}{2}\sqrt{2\alpha^2-10\alpha-1+2(\alpha+2)\sqrt{\alpha(\alpha+2)}}\right)n.$$
\end{proof}

\begin{figure}[h]
\centerline{\includegraphics[width=.5\textwidth]{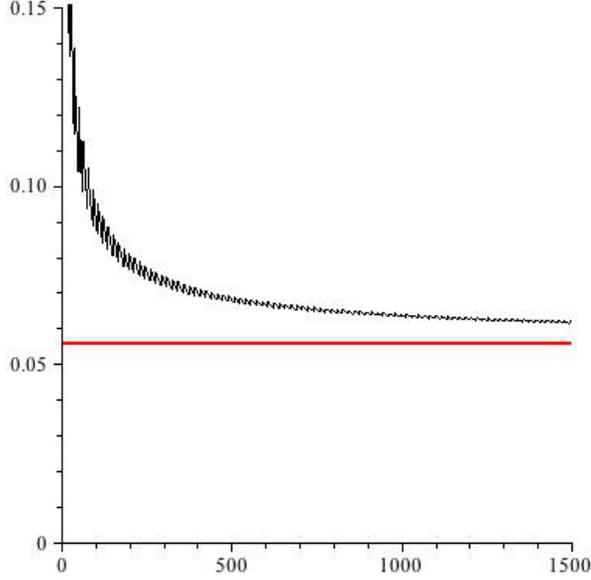}}
\caption{Comparison of \(\deg(\mathsf{HS}_{n,\lceil{n/.55}\rceil})/n\) (black) with its limit (red).
\label{fig:asympt}}
\end{figure}

Figure \ref{fig:asympt} shows the actual values of
\(\deg(\mathsf{HS}_{n,\lceil \alpha n\rceil})/n\) for
$\alpha=1/.55$. Notice that this sequence converges rather slowly. This is due to
the fact that we only take into account the first term in the
asymptotic expansion of \(\deg(\mathsf{HS}_{n,\lceil \alpha n\rceil})\). It would be possible to obtain the full
asymptotic expansion using techniques similar to those in
\cite{BarFauSal04,BarFauSal05}. However, this would not change the asymptotic complexity of Algorithm~\ref{algo:booleansolve}.

%% file: compl_alpha.tex
We now estimate the complexity of Algorithm~{\sf BooleanSolve} by
going through its steps and making all necessary hypotheses
explicit. We consider the case when the number of variables $n$ and
the number of polynomials $m$ are related by~$m\sim\alpha n$ for some
$\alpha\ge1$ and $n$ is large. Also we assume that the ratio $k/n$ is
controlled by a parameter $\gamma\in [0,1]$, i.e., $k=
(1-\gamma)n$.

The first step (lines 4 to 6 in the algorithm) is to evaluate the polynomials~$\tilde{f}_i$ from the polynomials~$f_i$. With no arithmetic operations, the polynomials~$f_i$ can first be written as polynomials in~$(x_1,\dots,x_{n-k})$ with coefficients that are polynomials of degree at most~2 in~$x_{n-k+1},\dots,x_n$ and at most~1 in each variable. Each such coefficient has at most $1+k+\binom{k}{2}$ monomials, each of which can be evaluated with at most one arithmetic operation. 
The total number of these polynomial coefficients is at most~$m(1+n-k+\binom{n-k}{2})$. Thus the total cost of all the evaluations of the coefficients of the polynomials~$\tilde{f}_i$ is at most~$O(n^52^{(1-\gamma) n})$. This turns out to be asymptotically negligible compared to the next steps.

The next stage (line~8) of our algorithm consists in performing tests of inconsistency of the Macaulay matrices.
\begin{prop}\label{prop:complFirstStep}For any $\epsilon>0$, $\alpha\ge 1$ and
  sufficiently large $m=\lceil\alpha n\rceil$, the complexity of all tests of
  consistency of Macaulay matrices in Algorithm
  \textsf{BooleanSolve} with parameters~$(m,n,k)$ is
        \begin{itemize}
                \item $O(2^{(1-\gamma+\theta F_{\alpha}(\gamma)+\epsilon)n})$ in the deterministic variant; 
                \item of expectation $O(2^{(1-\gamma+2 F_{\alpha}(\gamma)+\epsilon)n})$ in the probabilistic variant,
        \end{itemize}
where $\gamma=1-k/n$, $F_{\alpha}(\gamma)=-\gamma\log_2(D^D(1-D)^{1-D})$ with $D=M(\alpha/\gamma)$, the function $M$ as in Proposition~\ref{prop:dregh} and $\theta$ the complexity of linear algebra as in Proposition~\ref{prop:row_ech}.
\end{prop}
A notable feature of this result is that in terms of complexity, the probabilistic variant of our algorithm behaves as the deterministic one where the linear algebra would be performed in quadratic complexity (i.e., with $\theta=2$).
\begin{proof}
We first estimate the size of the Macaulay matrices. 
  By Proposition~\ref{prop:dregh}, the index~$d_0$, which
  is $1+\deg(\mathsf{HS}_{n-k,m})$ behaves
  asymptotically like $\gamma Dn$. The function $M(x)$ is decreasing
  for $x\ge 1$, so that $D\le M(1)<1/2$. Thus, $d_0<\gamma n/2$ for
  $n$ sufficiently large and Proposition~\ref{prop:sizemac} applies
  with~$d=d_0$, $m=\lceil\alpha n\rceil$ equations and $n-k=\gamma n$
  variables. For~$n$ sufficiently large, the bound
  for~$r_{\mathsf{Mac}}$ is larger than that
  for~$c_{\mathsf{Mac}}$, since 
  the quotient of these two bounds is
  ${m}/{(\frac {\gamma n} {d_0} - 1)^2}$, which grows linearly with~$n$.

Next, we turn to the tests of inconsistency. The previous bounds and Proposition~\ref{prop:row_ech} imply that
  the number of operations required for the computation of the row
  echelon form is
  $O(n\binom{\gamma n}{d_0}^{\theta})$. Similarly, by Proposition
  \ref{thm:wiedemann}, the complexity of checking the consistency of
  each matrix by the probabilistic method is
 $O(r_{\mathsf{Mac}}\log(r_{\mathsf{Mac}})s_{\mathsf{Mac}})=O(n^4\binom{\gamma n}{d_0}^2\log{\binom{\gamma n}{d_0}})$
  and that bound dominates the cost of the additional operations
  in~$\mathbb{F}_2$. Now, Stirling's formula implies that for
  any~$0<b<a$, $\log\binom{an}{bn}\sim
  n\log(a^a/(b^b(a-b)^{a-b}))$. Setting~$a=\gamma$ and $b=\gamma D$
  gives the result, the extra factor being due to the exhaustive
  search that performs this consistency check $2^{(1-\gamma)n}$ times.
\end{proof}
In the cases where the linear system $\mathbf{u}\cdot\mathsf{M}=\mathbf{r}$ is found inconsistent, then the polynomial system itself may be consistent and the algorithm proceeds with an exhaustive search (line~9) in a system with $\gamma n$ unknowns.
Each such search has cost $O(2^{(\gamma+\epsilon)n})$. As long as the number of these searches does not exceed $O(2^{(1-2\gamma+2F_{\alpha}(\gamma))n})$, the overall complexity of the algorithm is bounded by the complexity given in Proposition~\ref{prop:complFirstStep}. There can be two causes for the inconsistency of the linear system that triggers such a search: the existence of an actual solution with~$x_{n}=a_{n},\dots,x_{n-k+1}=a_{n-k+1}$; a witness degree of the specialized system larger than ${d}_0$ (e.g., if the homogenized specialized system is not boolean semi-regular). 
We now define a class of systems where this does not happen too much.

\begin{defn} \label{defn:strongsemireg}Let $S=(f_1,\ldots,f_m)$ be quadratic polynomials in $\mathbb{F}_2[x_1,\dots, x_n]$, $0\le k=(1-\gamma)n<n$, $\alpha= m/n$ and $d_0=1+\deg(\mathsf{HS}_{n-k,m})$. The system~$S$ is called $\gamma$-strong semi-regular if both the set of its solutions in~${\mathbb{F}_2^n}$ and the set              
\begin{multline*}
\bigl\{(a_{n-k+1},\dots, a_n)\in \mathbb F_2^k\mid \\ \dwit(f_1(x_1,\dots,x_{n-k},a_{n-k+1},\dots,a_n),\dots,f_m(x_1,\dots,x_{n-k},a_{n-k+1},\dots,a_n))> d_0\bigr\}
\end{multline*}
have cardinality at most ${2^{(1-2\gamma+2F_{\alpha}(\gamma))n}}$, with $F_{\alpha}$ as in Proposition~\ref{prop:complFirstStep}.
\end{defn}
Note that since $1-2\gamma+2F_{\alpha}(\gamma)$ is a decreasing function of~$\gamma$, a $\gamma$-strong semi-regular system is also $\gamma'$-strong semi-regular for any $\gamma'<\gamma$. 

The first condition for a system to be $\gamma$-strong semi-regular concerns its number of solutions. For boolean systems drawn uniformly at random, it is known that the probability that the number of boolean solutions is~$s$ decreases more than exponentially with~$s$ \citep{FusBac07}, so that the first condition is fulfilled with large probability. 
The second condition is related to the proportion of boolean semi-regular systems. We discuss this condition in the next section and show that it is also of large probability experimentally. Under this assumption of $\gamma$-strong semi-regularity, we now state the complexity of the algorithm obtained by optimizing the choice of the number~$k$ of variables that are specialized.

We first discuss large values of~$\gamma$. The function $1-2\gamma+2F_{\alpha}(\gamma)$ is decreasing with~$\alpha$
and negative when $\gamma=1$. Thus, the first condition implies that a $1$-strong semi-regular system has no solution. By continuity, this behavior persists for $\gamma$ close to~1 and actually holds for~$\gamma\in(0.824,1)$. It also persists for smaller values of $\gamma$ and larger $\alpha$.

\begin{cor}\label{cor:inconsistent}With the same notations as in Prop.~\ref{prop:complFirstStep}, when a system is $\gamma$-strong semi-regular with $\alpha$ and $\gamma$ such that $1-2\gamma+2F_{\alpha}(\gamma)<0$, then it is inconsistent and detected by Algorithm {\sf BooleanSolve} with parameters~$(m,n,0)$ in $O(2^{(\theta F_{\alpha}(1)+\epsilon)n})$ operations.
\end{cor}
\begin{figure}
  \centering
\input{graphTheta/exponent.tex}
  \caption{Exponent of the complexity for inconsistent systems in terms of the ratio $\alpha$ (see Thm. \ref{thm:globalcompl} and Cor.~\ref{cor:inconsistent})}
  \label{fig:exponent}
\end{figure}
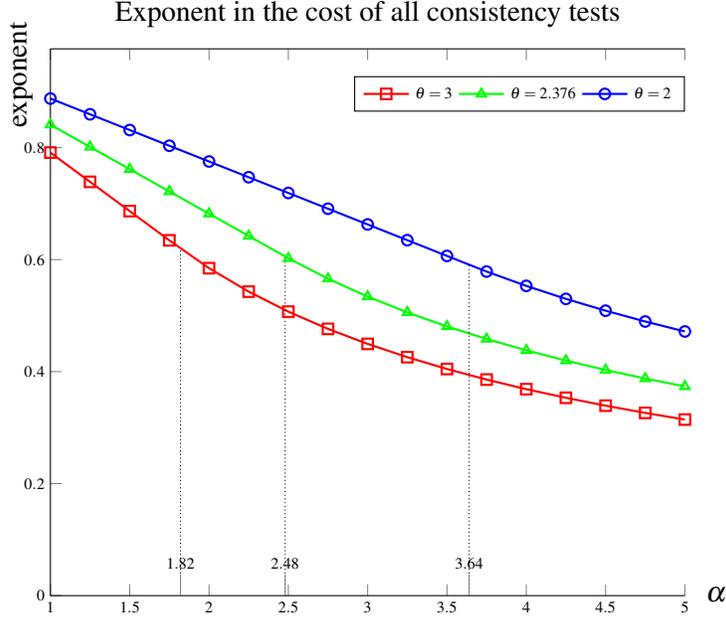
\noindent The value of the exponent
$\theta F_\alpha(1)$ in terms of $\alpha$ is plotted in
Figure~\ref{fig:exponent} (it corresponds to the right part of the plots, i.e. $\alpha>1.82$ for $\theta=2$, $\alpha>2.48$ for $\theta=2.376$, $\alpha>3.64$ for $\theta=3$).
\begin{proof}The hypothesis implies that the system has no solution and that its witness degree is bounded by~$d_0$, so that its absence of solution is detected by the linear algebra step in degree~$d_0$. In that case, no exhaustive search is needed. 
\end{proof}
 
For smaller values of~$\gamma$, the algorithm requires exhaustive searches. The optimal choice of~$k$ is obtained by an optimization on the complexity estimate. This leads to the following complexity estimates. In the next section, we argue that the required strong semi-regularities are very likely in practice, so that the only choice left to the user is that of the linear algebra routine.

\begin{thm}\label{thm:globalcompl}
Let $S=(f_1,\dots,f_m)$ be a system of quadratic polynomials in $\mathbb{F}_2[x_1,\dots, x_n]$, with $m=\lceil \alpha n\rceil$ and $\alpha\ge1$. 
Then Algorithm~\textsf{BooleanSolve} finds all its roots in $\mathbb{F}_2^n$ with a number of arithmetic operations in~$\mathbb{F}_2$ that is
\begin{itemize}
        \item $O(2^{(1-0.112\alpha) n})$ if $S$ is $(.27\alpha)$-strong semi-regular using Gaussian elimination for the linear algebra step;
        \item $O(2^{(1-0.159\alpha) n})$ if $S$ is $(.40\alpha)$-strong semi-regular using computation of the row echelon form with Coppersmith-Winograd multiplication;
        \item of expectation $O(2^{(1-0.208\alpha) n})$ if $S$ is $(.55\alpha)$-strong semi-regular using the probabilistic Algorithm~\ref{algo:testinconsistency}.
\end{itemize}
In all cases, the value of~$k$ passed to the algorithm is $\lceil n(1-\gamma)\rceil$ with $\gamma$ corresponding to the strong semi-regularity.
\end{thm}
\begin{proof}The correctness of the algorithm has already been proved in Proposition~\ref{algo_correct}. Only the complexity remains to be proved.

  By definition of strong semi-regularity, the number of exhaustive
  searches that need be performed in line~9 of the Algorithm
  is~$O(2^{(1-2\gamma+2F_{\alpha}(\gamma))n})$, each of them
  using~$O(2^{(\gamma+\epsilon)n})$ arithmetic operations for
  any~$\epsilon>0$. It follows that the overall cost of these
  exhaustive searches
  is~$O(2^{(1-\gamma+2F_{\alpha}(\gamma)+\epsilon)n})$; it is
  bounded by the cost of the tests of inconsistency. We now choose
  $\gamma$ in such a way as to minimize this cost, in terms of
  $\alpha$. Direct computations lead to the following numerical
  results, that conclude the proof.
\end{proof}

\begin{lem} With the same notation as in Proposition~\ref{prop:complFirstStep}, the function $1-\gamma+\theta F_{\alpha}(\gamma)$ is bounded by
        \begin{itemize}
        \item $1 - 0.112\alpha $ when $\theta=3$ and $\gamma=0.27\alpha$;
        \item $1 - 0.159\alpha$ when $\theta=2.376$ and $\gamma=0.40\alpha$;
                \item $1 - 0.208\alpha$ when $\theta=2$ and $\gamma=0.55\alpha$.
        \end{itemize}
\end{lem}
\begin{proof}
The function $1-\gamma + \theta F_{\alpha}(\gamma)$ has two parameters but its extrema can be found by reducing it to a one parameter function. Indeed, this function is maximal for $\alpha\ge1$ and $\gamma\in[0,1]$ when $(-\gamma+\theta F_\alpha(\gamma))/\alpha$ is. Setting $\lambda=\gamma/\alpha$, this is exactly $-\lambda+\theta F_1(\lambda)$, with $\lambda\in[0,1/\alpha]$.
Direct computations lead to the optimal $\lambda$'s:
$\lambda=\min(1/\alpha,0.27)$ when $\theta=3$, $\lambda=\min(1/\alpha,0.40)$ when $\theta=2.376$, $\lambda=\min(1/\alpha,0.55)$ when $\theta=2$.
\end{proof}

%% file: graphTheta/exponent.tex
\begin{tikzpicture}
\pgfplotsset{every axis legend/.append style={
    at={(1,0.95)},
anchor=north east}}
\begin{axis}[
scale only axis,
xmin=1,ymin=0, xmax=5,
axis y line*=left,
xlabel style={at={(1.05,0.05)},anchor=south}, 
ylabel style={at={(0.1,1.06)},anchor=east}, 
xlabel={\normalsize $\alpha$},
legend columns=3, 
ylabel={\normalsize exponent},
title={\normalsize Exponent in the cost of all consistency tests}]

\addplot[red, mark=square,thick] coordinates{
(1.0000, 0.7911)
(1.2500, 0.7388)
(1.5000, 0.6866)
(1.7500, 0.6344)
(2.0000, 0.5847)
(2.2500, 0.5428)
(2.5000, 0.5071)
(2.7500, 0.4763)
(3.0000, 0.4495)
(3.2500, 0.4258)
(3.5000, 0.4047)
(3.7500, 0.3858)
(4.0000, 0.3687)
(4.2500, 0.3533)
(4.5000, 0.3392)
(4.7500, 0.3263)
(5.0000, 0.3144)
};
\addplot[green, mark=triangle,thick] coordinates{
(1.0000, 0.8410)
(1.2500, 0.8012)
(1.5000, 0.7615)
(1.7500, 0.7217)
(2.0000, 0.6820)
(2.2500, 0.6422)
(2.5000, 0.6024)
(2.7500, 0.5659)
(3.0000, 0.5340)
(3.2500, 0.5058)
(3.5000, 0.4807)
(3.7500, 0.4583)
(4.0000, 0.4381)
(4.2500, 0.4197)
(4.5000, 0.4030)
(4.7500, 0.3877)
(5.0000, 0.3736)
};
\addplot[blue, mark=o,thick] coordinates{
(1.0000, 0.8876)
(1.2500, 0.8595)
(1.5000, 0.8314)
(1.7500, 0.8033)
(2.0000, 0.7752)
(2.2500, 0.7471)
(2.5000, 0.7190)
(2.7500, 0.6909)
(3.0000, 0.6628)
(3.2500, 0.6347)
(3.5000, 0.6066)
(3.7500, 0.5787)
(4.0000, 0.5531)
(4.2500, 0.5299)
(4.5000, 0.5088)
(4.7500, 0.4895)
(5.0000, 0.4717)
};
\addplot+[black,densely dotted, mark=none, ycomb] plot coordinates {(1.82,0.6198) (2.48,0.6056) (3.64,0.5908)};
\node[coordinate,pin=above:{$1.82$}] at (axis cs:1.82,0){};
\node[coordinate,pin=above:{$2.48$}] at (axis cs:2.48,0){};
\node[coordinate,pin=above:{$3.64$}] at (axis cs:3.64,0){};
% \addplot[black,densely dotted, mark=none] coordinates {(2.48,0) (2.48,0.6056)};
% \addplot[black,densely dotted, mark=none] coordinates {(3.64,0) (3.64,0.5908)};
\tiny\legend{$\theta=3$,$\theta=2.376$, $\theta=2$}
\end{axis}

% \pgfplotsset{every axis legend/.append style={ at={(1.1,-0.1)},
%     anchor=north east}}
% \normalsize
% \begin{semilogyaxis}[scale only axis,xmin=10,xmax=24,axis y line*=right,axis x
%   line=none,log basis y=2, legend columns=4, ylabel={\normalsize abs. value of the first non-pos. coeff.}]
% \addplot[gray,dashed] table[x=n,y=coeffHS]{../StrongReg0.450};
% \tiny\legend{coeffHS}
% \end{semilogyaxis}

\end{tikzpicture}

%% file: picture.tex
\section{Numerical Experiments on Random Systems}
\label{sec:expmodel}

\textbf{Probabilistic model.} In this section, we study experimentally the behavior of Algorithm
\textsf{BooleanSolve} of random quadratic systems where each
coefficient is 0 or 1 with probability 1/2. These random boolean quadratic
systems appear naturally in Cryptology since the security of several
recent cryptosystems relies directly on the difficulty of solving such
systems (see e.g., \cite{QUADEuro,QUADJSC}).

\subsection{$\gamma$-strong semi-regularity}
\label{sec:numBeta}

The goal of this section is to give experimental evidence that the
assumption of $\gamma$-strong semi-regularity is not a strong condition
for random boolean systems. This is related to the notoriously difficult
conjecture by \cite{Fro85}, which states that in characteristic $0$,
almost all systems are semi-regular (with the meaning of
semi-regularity given in \citep{BarFauSal05}), see also \citep{Moreno-Socias2003}.

Consequently, we propose the following conjecture, which can be seen
as a variant of Fr\"oberg's conjecture for boolean systems:

\begin{conj}\label{conj:propbeta}For any~$\alpha\ge 1$ and $\gamma<1$ such that
  $1-2\gamma+2F_{\alpha}(\gamma)>0$, the proportion of $\gamma$-strong
  semi-regular systems of $\lceil\alpha n\rceil$ quadratic
  polynomials in~$\mathbb{F}_2[x_1,\dots,x_n]$ tends to~1
  when~$n\rightarrow\infty$.
\end{conj}

The rest of this section is devoted to providing experiments supporting this conjecture.

\begin{figure}
\centerline{
\fbox{\includegraphics[width=.5\textwidth]{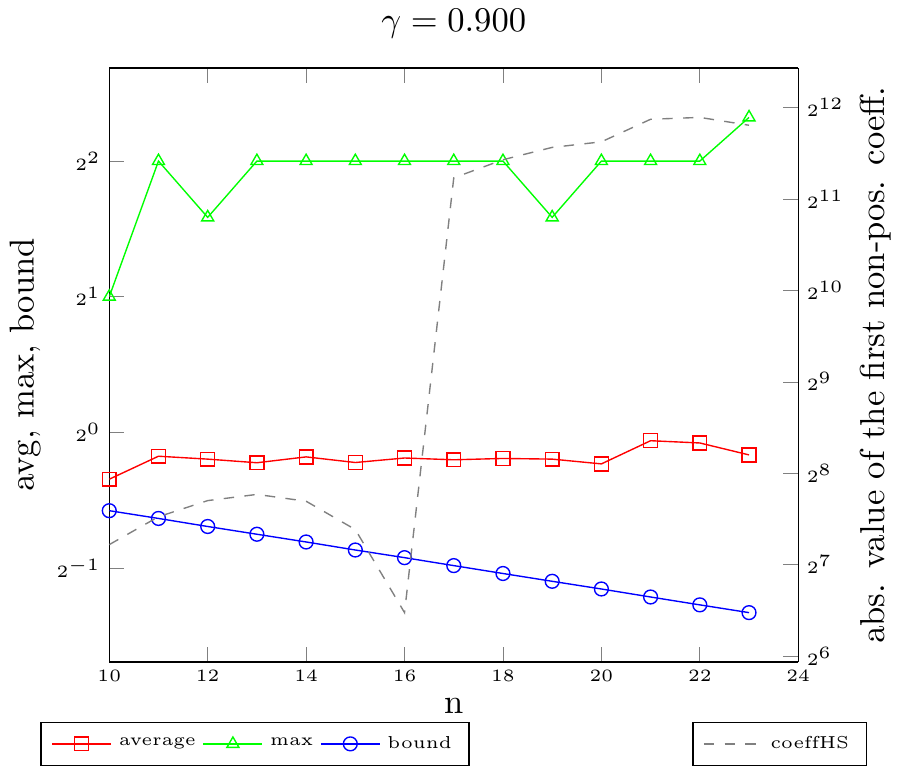}}
\fbox{\includegraphics[width=.5\textwidth]{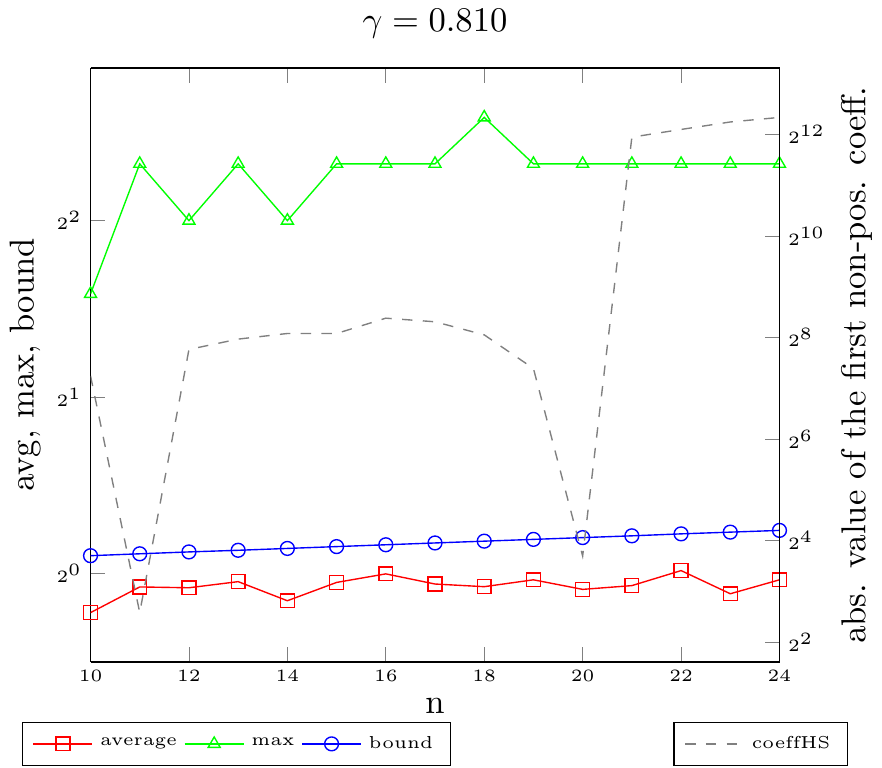}}}
\centerline{
\fbox{\includegraphics[width=.5\textwidth]{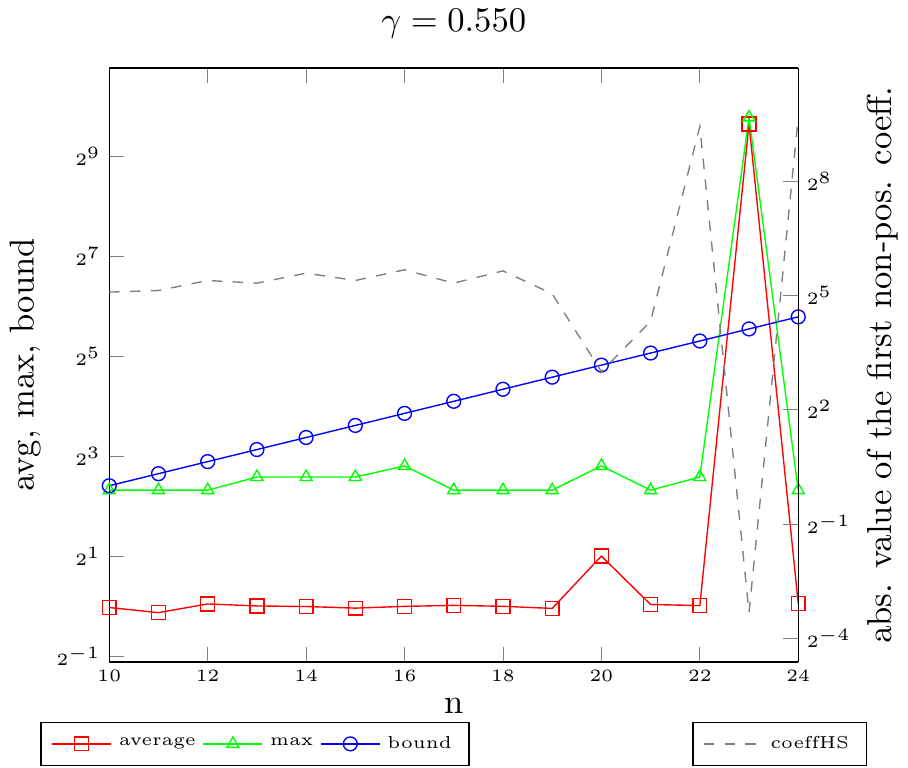}}
\fbox{\includegraphics[width=.5\textwidth]{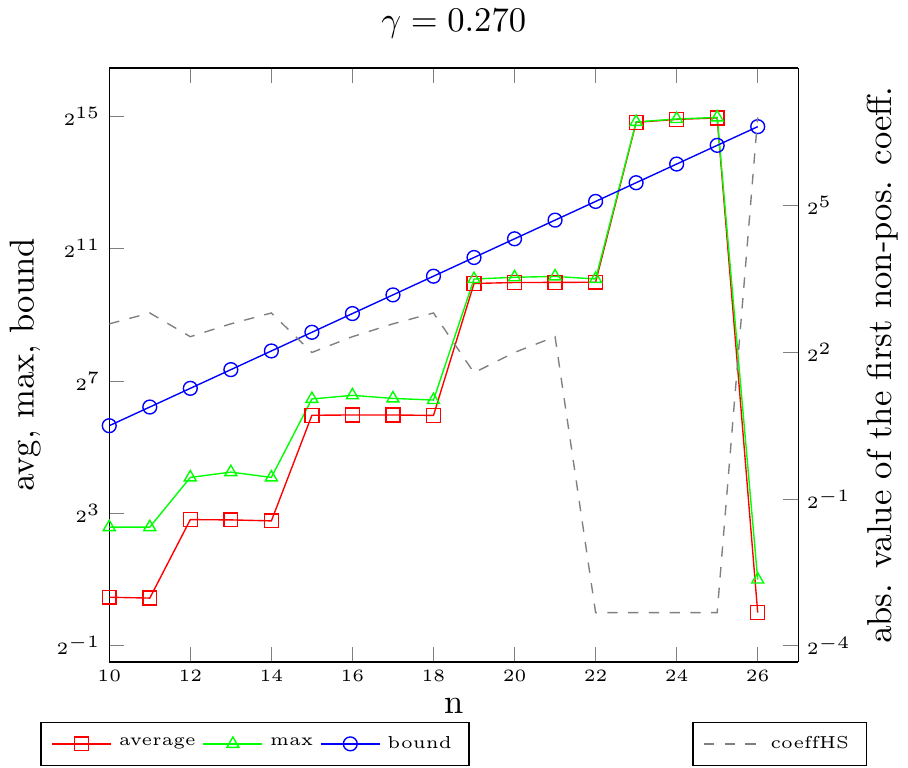}}}
\centerline{
\fbox{\includegraphics[width=.5\textwidth]{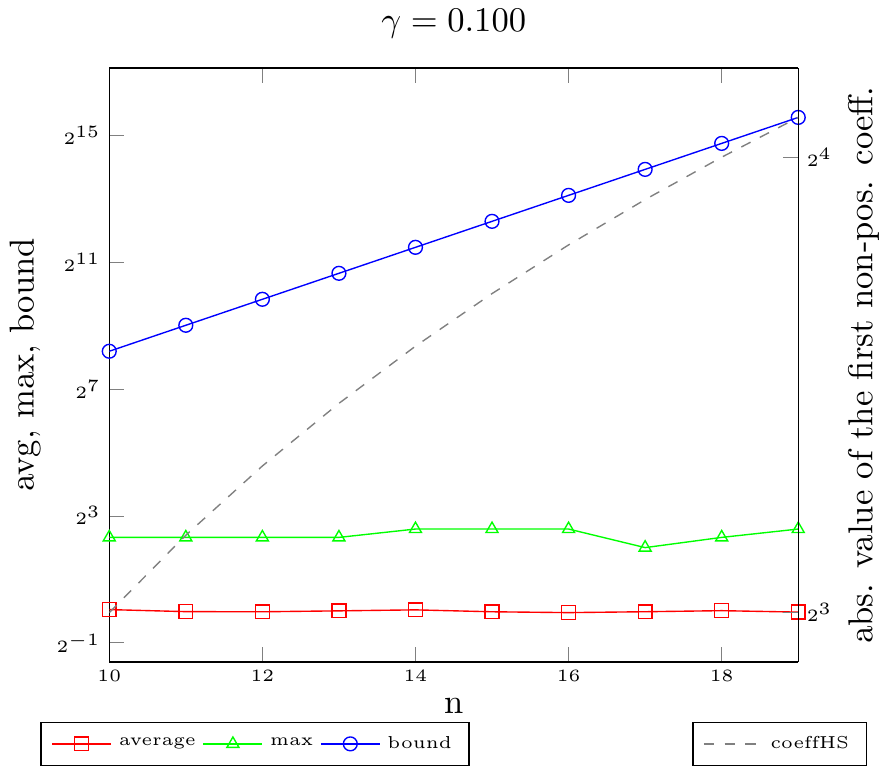}}}
\caption{Relation between the quality of the filtering, the value of
  the first nonpositive coefficient of $\mathsf{HS}_{\lfloor\gamma
    n\rfloor,n}$, and $\gamma$-strong
  semi-regularity.  In red
  (resp. green), the average (resp. maximum) number of specializations
  for which the linear system is inconsistent. In blue, the bound for
  $\gamma$-strong regularity. Dashed line: absolute value of the first
  non positive coefficient of $\mathsf{HS}_{\lfloor\gamma
    n\rfloor,n}$. \label{fig:strongsemireg}}
\end{figure}

In Figure \ref{fig:strongsemireg}, we show the relation between the
value of the first nonpositive coefficient of the power series
expansion of $\mathsf{HS}_{\lfloor \gamma n\rfloor,n}$ and
$\gamma$-strong semi-regularity for small values of $n=m$
(i.e. $\alpha=1$). For each $n$, the experiments are conducted on 1000
random quadratic boolean systems. For each of these systems, we
compute the $2^{\lceil (1-\gamma) n\rceil}$ specialized systems and we count
the number of specializations for which the filtering linear system is
inconsistent.

Four curves are represented on each chart in Figure
\ref{fig:strongsemireg}. The red (resp. green) one represents the
average (resp. maximal) number of specializations for which the linear
system (step 8 of Algorithm \textsf{BooleanSolve}) is
inconsistent. In contrast, the blue curve shows the upper bound on this
number of specializations required to be $\gamma$-strong semi-regular
(see Definition \ref{defn:strongsemireg}). The black curve shows the
absolute value of the first nonpositive coefficient of the
corresponding power series (i.e. $\mathsf{HS}_{\lfloor\gamma
  n\rfloor,n}$). The $y$-axis is represented in logarithmic scale. The
value $\gamma=0.1$ is never used in the complexity analysis (since in
Theorem \ref{thm:globalcompl}, $\gamma\geq .27$ for any value of
$\alpha\geq 1$). However, it is still interesting to study the
behavior of Algorithm \ref{algo:booleansolve} when almost all
variables are specialized: the filtering remains very efficient in
this case, and the branches which are explored during the second stage
of the exhaustive search correspond to those containing solutions of
the system. 

\medskip

\noindent \textbf{Interpretation of Figure \ref{fig:strongsemireg}.}
First, notice that for $\gamma\leq 0.55$ the green curve is always below
the blue one (except for the case $\gamma=.55, n=23$), meaning that
during our experiments, all randomly generated systems with those
parameters were $\gamma$-strong semi-regular.

Next, in most curves (except $\gamma=0.27$), the average (resp. maximal)
number of points where the specialization leads to an inconsistent
linear system is close to $1$ (resp. $5$). This can be explained by a
simple Poisson model.  Indeed, the
number of solutions of a random boolean system with as many equations
as unknowns follows a Poisson law with parameter $1$ (see~\cite{FusBac07}). Therefore, the
expectation of the number of solutions is $1$. The
expectation of the maximum of the number of solutions of $1000$ random
systems is then given as the maximum of 1000~iid random variables $P_1,\ldots, P_{1000}$ following a Poisson law of parameter $1$: 
\[\mathbf{E}(\max(P_1,\ldots, P_{1000}))=\sum_{k\ge1}{k\left((e^{-1}\sum_{i=0}^k{\frac{1}{i!}})^{1000}-(e^{-1}\sum_{i=0}^{k-1}{\frac{1}{i!}})^{1000}\right)}\simeq 5.51,\]
which explains very well the observed behaviour.

This means that during Algorithm
\ref{algo:booleansolve} with these parameters, almost all
specializations giving rise to an inconsistent system correspond to a
branch of the exhaustive search which contains an actual solution of
the system. Therefore, the filtering is very efficient for those
parameters. 

\medskip

\noindent\textbf{Few specializations.} In the case $\gamma=0.9$, the blue curve has a negative slope. This is due to
the fact that the quantity $1-2\gamma+2F_{\alpha}(\gamma)$ (see
Definition \ref{defn:strongsemireg}) is negative for $\alpha=1$ and
$\gamma>0.82308$. Therefore, we cannot expect that a large proportion
of boolean systems are $\gamma$-strong semi-regular in this setting. A
limit case is investigated in the chart corresponding to
$\gamma=0.81$. There, $1-2\gamma+2F_{\alpha}(\gamma)\approx 0.0102$ is
positive but very close to zero. Experiments show that random boolean
systems with these parameters and $10\leq n\leq 24$ are $\gamma$-strong
semi-regular with probability approximately equal to $0.75$.

\medskip

\noindent\textbf{Absolute value of the first nonpositive coefficient of
  $\mathsf{HS}_{\lfloor\gamma n\rfloor,n}$ and $\gamma$-strong
  semi-regularity.} Another interesting setting is $\gamma=.55,
n=23$. Here, no generated systems were $\gamma$-strong semi-regular
(although all generated systems for $n\neq 23$ were $\gamma$-strong
semi-regular). As explained in Section \ref{sec:proba}, this is due to
the fact that the first nonpositive coefficient of the power series
expansion of $\mathsf{HS}_{\lfloor\gamma n\rfloor,n}$ is equal to
zero. In Section~\ref{sec:BoolDelta}, we show that this phenomenon
can be avoided by a simple variant of the algorithm.

A similar phenomenon happens for $\gamma=.27$: the first nonpositive
coefficient of the power series has small absolute value. It is an
accident due to the fact that this coefficient is close to zero for
$n\leq 25$ (see Figure \ref{fig:HScoeff}). On this chart, we can see
clearly the relation between the absolute value of the first
nonpositive coefficient of $\mathsf{HS}_{\lfloor\gamma n\rfloor,n}$
and the number of specializations for which the consistency test
fails.

Indeed, experiments on $1000$ random systems with $\gamma=.27$ and
$n=26$ were conducted and in this case the average number of
specializations for which the linear system is inconsistent is
$1$.

These experiments justify the fact that the complexity analysis
conducted in Section~\ref{sec:compl} is relevant for a large class of
boolean systems. Also, it shows that the random systems for which the
filtering may not be efficient can be detected \emph{a priori} by
looking at the absolute value of the first nonpositive coefficient in
the power series. If this value is small, we show in Section
\ref{sec:BoolDelta} that the quality of the filtering can be improved at low
cost by adding redundancy.

\begin{figure}
\centerline{
\fbox{\includegraphics[width=.5\textwidth]{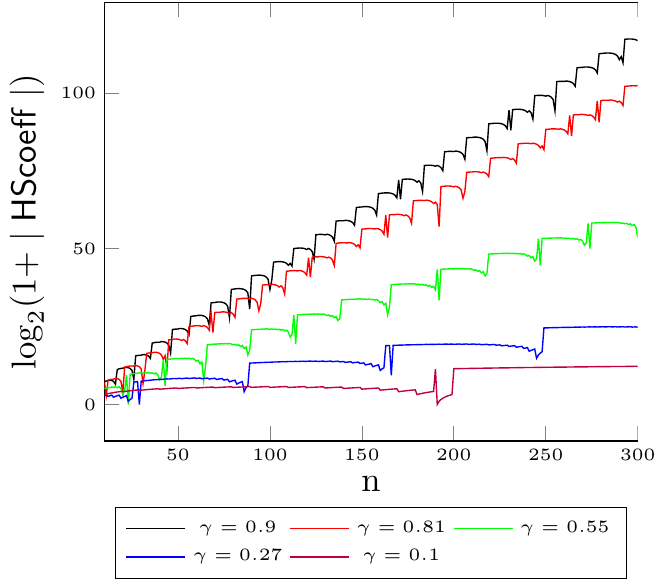}}}
  \caption{Evolution of the logarithm of the absolute value of the first nonpositive coefficient of $\mathsf{HS}_{\lfloor\gamma n\rfloor,n}$. \label{fig:HScoeff}}
\end{figure}

Figure \ref{fig:HScoeff} shows the evolution of the logarithm of the
absolute value of the first nonpositive coefficient of
$\mathsf{HS}_{\lfloor\gamma n\rfloor,n}$. This absolute value seems to
grow exponentially with $n$ for any given $\gamma$.
Since the quality of the filtering is related to this absolute value, these experiments suggest that the proportion of $\gamma$-strong semi-regular systems tends towards 1 when $n$ grows, as formulated in Conjecture \ref{conj:propbeta}.

\subsection{Numerical estimates of the complexity}
When $n=m$ and in the most favorable algorithmic case, our complexity estimate
uses~$\gamma=.55$. For this value, we display in
Figure~\ref{fig:asympt} (page \pageref{fig:asympt}) a comparison of
the behaviour of~$\deg(\mathsf{HS}_{n,\lceil{\frac{n}{\gamma}}\rceil})/n$ and
its limit. 
This picture shows a relatively slow
convergence.  Thus, for a given number~$n$ of variables it is more
interesting to optimize~$\gamma$ using the exact value
of~$\deg(\mathsf{HS}_{\lfloor{\gamma n}\rfloor,n})$ rather than a first order asymptotic estimate. In the same
spirit, one can also use the actual values given by Eq.~\eqref{eq:mac}
for the Macaulay matrix. Thus we seek to find~$\gamma$ that minimizes
the following bounds on the number of operations:
\begin{equation}\label{eq:nbops}
\begin{array}{l}
2^{(1-\gamma) n} r_{\mathsf{Mac}}c_{\mathsf{Mac}}\min(r_{\mathsf{Mac}},c_{\mathsf{Mac}})^{\theta-2}, \\\text{resp.}\quad 2^{(1-\gamma) n}\max(r_{\mathsf{Mac}},c_{\mathsf{Mac}})\log\max(r_{\mathsf{Mac}},c_{\mathsf{Mac}})s_{\mathsf{Mac}}
\end{array}\end{equation}
in the deterministic (resp. probabilistic) variants, using
Eq.~\eqref{eq:size_matrix} with~$n$ equations, $\lfloor{\gamma
  n}\rfloor$ variables and~$d=\deg(\mathsf{HS}_{\lfloor{\gamma
    n}\rfloor,n})$. The corresponding values of~$\gamma$ are given in
Figure~\ref{fig:beta}, together with the corresponding values of the
quantities in Eq.~\eqref{eq:nbops}. Although these values do not take
into account the constants hidden in the $O()$ estimates of the
complexity, they suggest the relevance of these algorithms in the
cryptographic sizes: the threshold between exhaustive search and our
algorithm with Gaussian elimination is~$n\simeq280$, while the
asymptotically faster Las Vegas variant starts being faster than
exhaustive search for~$n$ larger than~200 and beats deterministic
Gaussian elimination for $n$ larger than~160.

\begin{figure}
\centerline{\includegraphics[width=.4\textwidth]{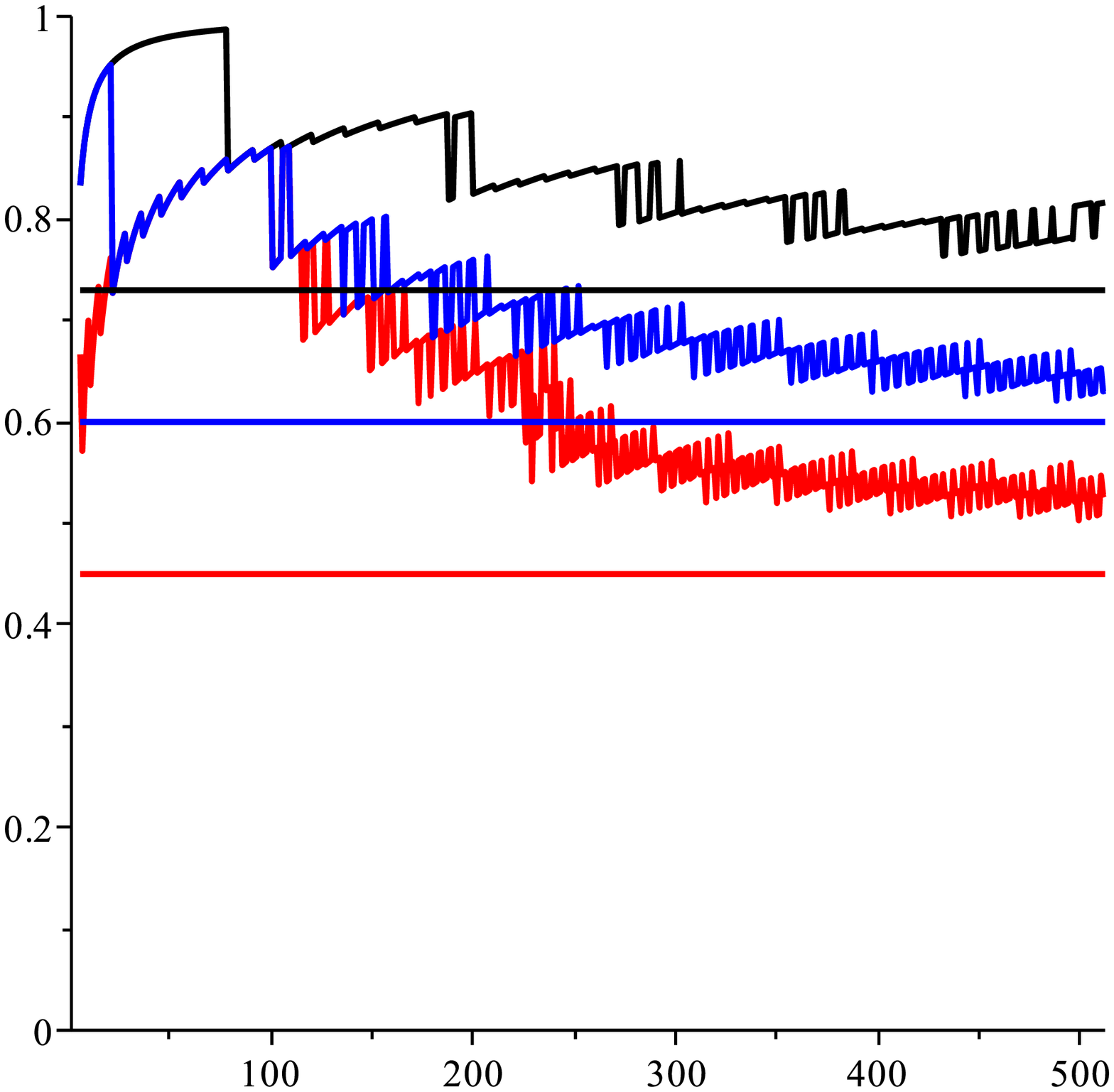}
\hfil\includegraphics[width=.55\textwidth]{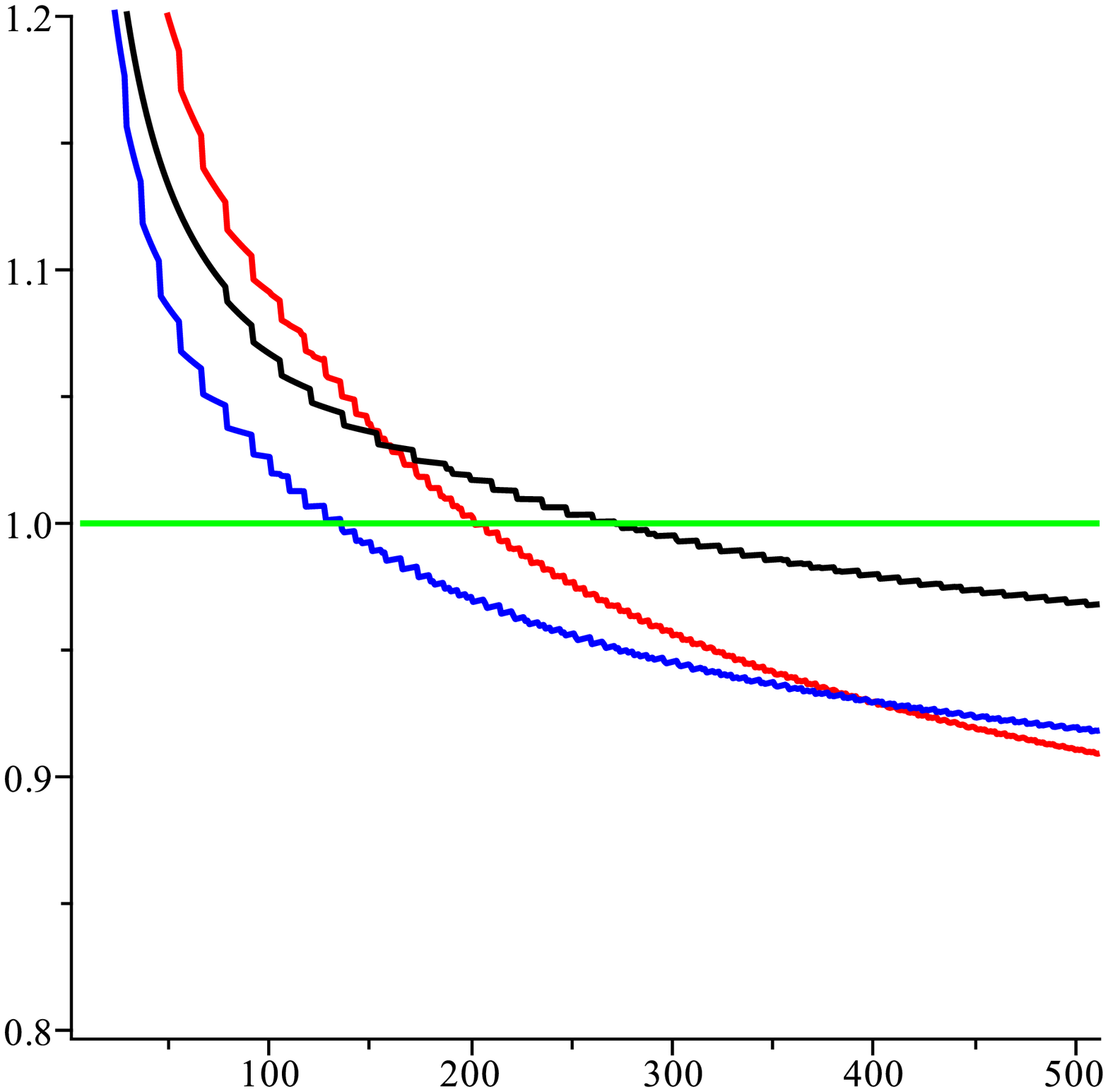}}
\caption{Left: optimal values of~$\gamma$ for the probabilistic 
variant (red), the deterministic variant with Gaussian
 elimination (black) and Coppersmith-Winograd matrix
  multiplication (blue), and their limits. Right: 
  corresponding values of $\log_2N/n$, with $N$ given 
  by~\eqref{eq:nbops}. The green line corresponds to an 
  exhaustive search.\label{fig:beta}} 
\end{figure}

%% file: proba.tex
\begin{figure}
\centering
\begin{minipage}{.8\textwidth}
\centerline{\includegraphics[width=.7\textwidth]{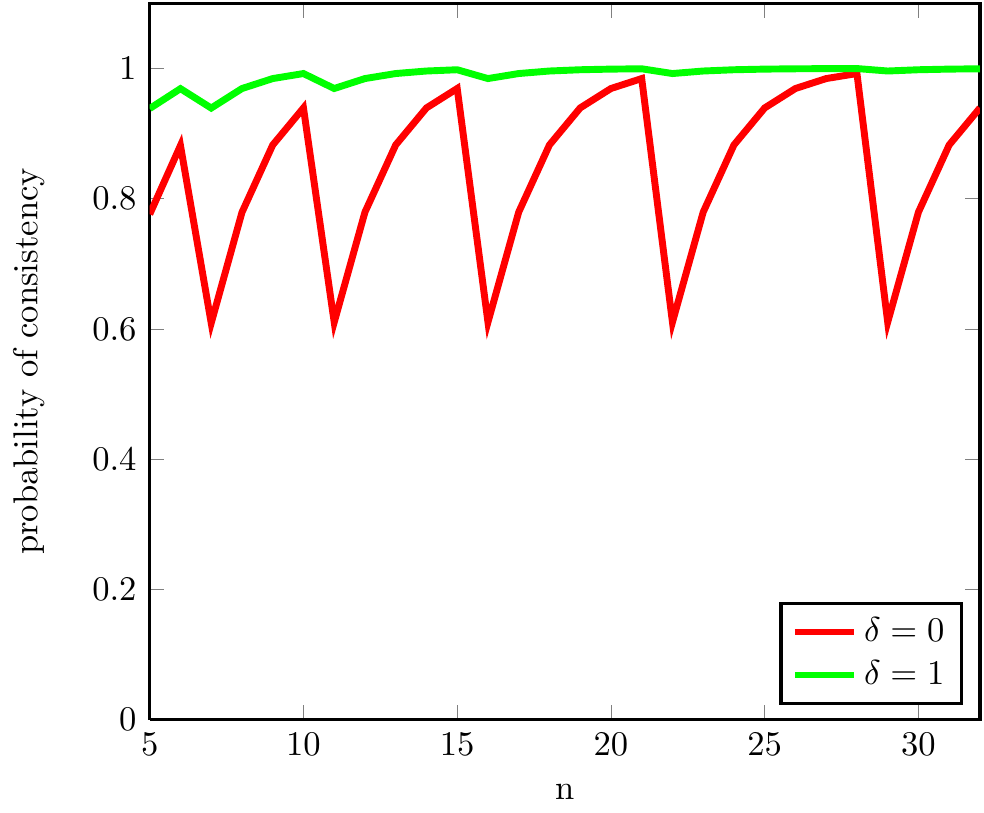}}
\caption{Proportion of specialized quadratic systems for which the
  linear system (line 9 of Algorithm \ref{algo:booleansolve}) is
  consistent. Parameters: $k=\left\lceil 1/2+n-\frac{\sqrt{-7+8
        n}}{2}\right\rceil$. In red, $\delta=0$ (corresponding to
  Algorithm \ref{algo:booleansolve}); in green, $\delta=1$ (see
  Algorithm \ref{algo:booleansolvev} of Section \ref{sec:BoolDelta}).
\label{fig:proba}}
\end{minipage}
\end{figure}

We showed in Section \ref{sec:numBeta} that when the first nonpositive
coefficient of $\mathsf{HS}_{n-k,n}$ is close to zero, then the linear
filtering may not be as efficient as expected (for instance in the
case $\gamma=.55$, $n=23$ in Figure \ref{fig:strongsemireg}). Another
case is shown in Figure \ref{fig:proba}. The curve $\delta=0$ shows
the behavior of Algorithm \ref{algo:booleansolve} on random square
systems ($m=n$) where $k$ is chosen as small as possible such that the
witness degree is $\dwit=2$: this is obtained by choosing
$k=\left\lceil 1/2+n-\frac{\sqrt{-7+8 n}}{2}\right\rceil$ (that is
$d_0=2$).

First, we observe that specializing a uniformly distributed random quadratic polynomial
$P\in\mathbb F_2[x_1,\ldots, x_n]$ at a uniformly distributed random point in $\mathbb
F_2^k$ yields a random polynomial that is also uniformly distributed 
in $\mathbb F_2[x_1,\ldots, x_{n-k}]$. We assume here that $P$ is reduced modulo the field equations $\langle
x_1^2-x_1,\ldots, x_n^2-x_n\rangle$. Let us assume first that
$k=1$. Then $P$ can be rewritten as
$$P(x_1,\ldots,x_n)=x_n P_1(x_1,\ldots, x_{n-1})+ P_2(x_1,\ldots, x_{n-1}),$$
where $P_1$ (resp. $P_2$) is a random polynomial following a uniform
distribution on the set of reduced boolean polynomials of degree $1$
(resp. of degree $2$) in $\mathbb F_2[x_1,\ldots, x_{n-1}]$.
Therefore, if $a\in \mathbb F_2$ is a random variable, $P(x_1,\ldots,
x_{n-1},a)\in \mathbb F_2[x_1,\ldots, x_{n-1}]$ is either $P_1$ or
$P_1+P_2$ and thus follows a uniform distribution on the set of
reduced quadratic boolean polynomials. The extension to arbitrary~$k<n$ follows by induction.

Consequently, in the special case $d_0=2$ of Figure \ref{fig:proba}
the boolean Macaulay matrix of a specialized system will be uniformly
distributed among the boolean matrices with the same dimensions. Also,
due to the choice of $k$, it will be roughly square. However, in
$\mathbb F_2$, the probability that a random square matrix has full
rank is not close to $1$. An estimate of this probability can be
obtained as follows.

The probability that a
random $p\times q$ boolean matrix has rank $r$ is (see \cite{FisherAlexander1966,Stitzinger87})
$$P(p,q,r)=2^{-p q}\displaystyle\frac{\prod_{j=0}^{r-1}(2^p-2^j)\prod_{j=0}^{r-1}(2^q-2^j)}{\prod_{j=0}^{r-1}(2^r-2^j)}.$$

Therefore, given a nonzero vector $\mathbf v\in \mathbb F_2^p$ and a random boolean $p\times q$ matrix $\mathsf M$, the probability that the linear system $\mathbf u\cdot \mathsf M = \mathbf v$ is consistent is
$$Q(p,q)=\sum_{i=1}^{p} P(p,q,i)\left(\frac{2^{i}-1}{2^q-1}\right).$$

Direct numerical computations show that for square matrices,
$Q(p,p)\approx 0.61$ as soon as $p\geq 4$. This probability
corresponds to the valleys of the curve $\delta=0$ in Figure~\ref{fig:proba}.  Also,
it can be noticed that $Q(p,q)$ grows quickly with $p-q$. For
instance, $Q(p+6,p)\approx 0.99$ when $p\geq 1$.
 
Consequently, it is interesting to specialize more variables than $k$ in
some cases (especially when the first nonpositive coefficient of ${(1+t)^{n-k}}/({(1-t)(1+t^2)^{m}})$ has small
absolute value): doing so increases the difference between the
dimensions of the Macaulay matrices. This does not change the
correctness of the algorithm (nor its asymptotic complexity), but
increases the effectiveness of the filtering performed by linear
algebra.

\subsection{Improving the quality of the filtering for small values
of \(n\)}
\label{sec:BoolDelta}
In this section, we propose an extension of Algorithm
\textsf{BooleanSolve} which takes an extra argument $\delta$, in order
to avoid the behavior of the algorithm shown in Section~\ref{sec:proba}. The main idea is to specialize $k+\delta$ variables,
but to take only $k$ into account for the computation of
$d_0$. Consequently, the difference between the number of columns and
the rank of the Macaulay matrix is not too small, and hence the linear
filtering performs better.
The resulting algorithm is given in Algorithm \ref{algo:booleansolvev}.

\begin{algorithm}
\caption{improved BooleanSolve.}
\label{algo:booleansolvev}
\begin{algorithmic}[1]
\Require $m, n, k, \delta\in\mathbb N$ such that $k+\delta<n\leq m$ and $f_1,\ldots,f_m$ quadratic polynomials in $\mathbb{F}_2[x_1,\ldots, x_n]$.
\Ensure The set of boolean solutions of the system $f_1=\dots=f_m=0$.
\State $S:=\emptyset$. 
\State ${d_{0}}:=$ index of the first nonpositive coefficient in the series expansion of the rational function $\frac{(1+t)^{n-k}}{(1-t)(1+t^2)^{m}}$.
\ForAll{$(a_{n-k-\delta+1},\ldots, a_n)\in\mathbb F_2^{k+\delta}$}
\For{$i$ from $1$ to $m$}
\State $\tilde{f}_i(x_1,\ldots,x_{n-k-\delta}):=f_i(x_1,\ldots, x_{n-k-\delta},a_{n-k-\delta+1},\ldots, a_n)$.
\EndFor 
\State $\mathsf{M}:=$ boolean Macaulay matrix of $(\tilde{f}_1,\ldots,\tilde{f}_m)$ in degree ${d_{0}}$.
\If{the system $\mathbf u\cdot\mathsf M=\mathbf r$ is inconsistent}\Comment{$\mathbf r$ as defined in Lemma~\ref{lem:inconsistent}}
\State $T:=$ solutions of the system ($\tilde{f}_1=\cdots=\tilde{f}_m=0$) (exhaustive search).
\ForAll{$(t_1,\ldots, t_{n-k-\delta})\in T$}
\State $S:=S\cup \{(t_1,\ldots, t_{n-k-\delta},a_{n-k-\delta+1},\ldots, a_n)\}$.
\EndFor
\EndIf
\EndFor
\State \Return $S$.
\end{algorithmic}
\end{algorithm}

In Figure \ref{fig:proba},
 we show the role of the parameter $\delta$
when $k$ is chosen minimal such that $d_0=2$: adding redundancy by
choosing a nonzero $\delta$ can greatly improve the quality of the
filtering (in practice, choosing $\delta=1$ is sufficient).

\begin{figure}[H]
\centerline{
\fbox{\includegraphics[width=.5\textwidth]{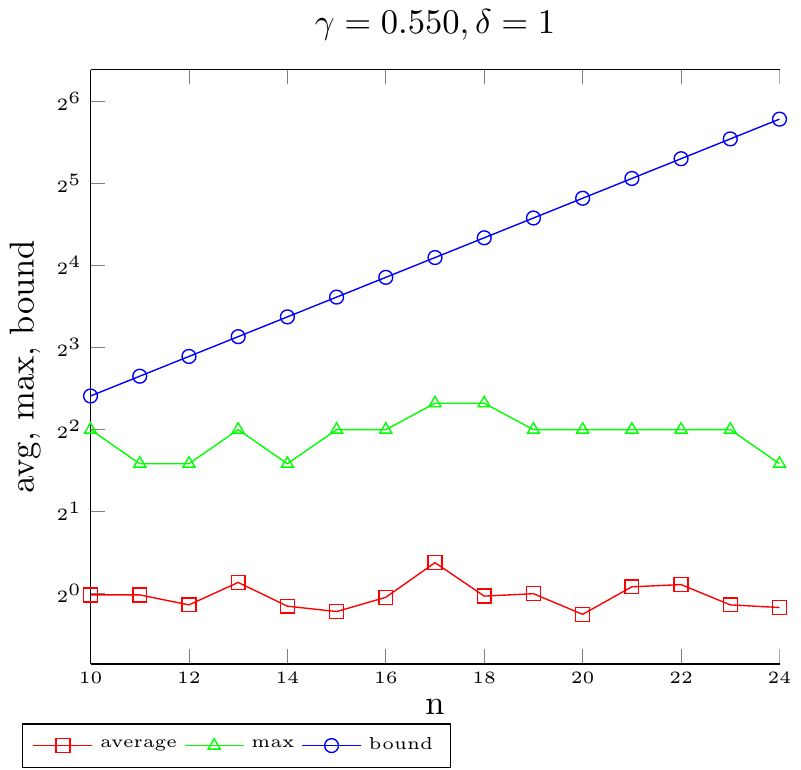}}
\fbox{\includegraphics[width=.5\textwidth]{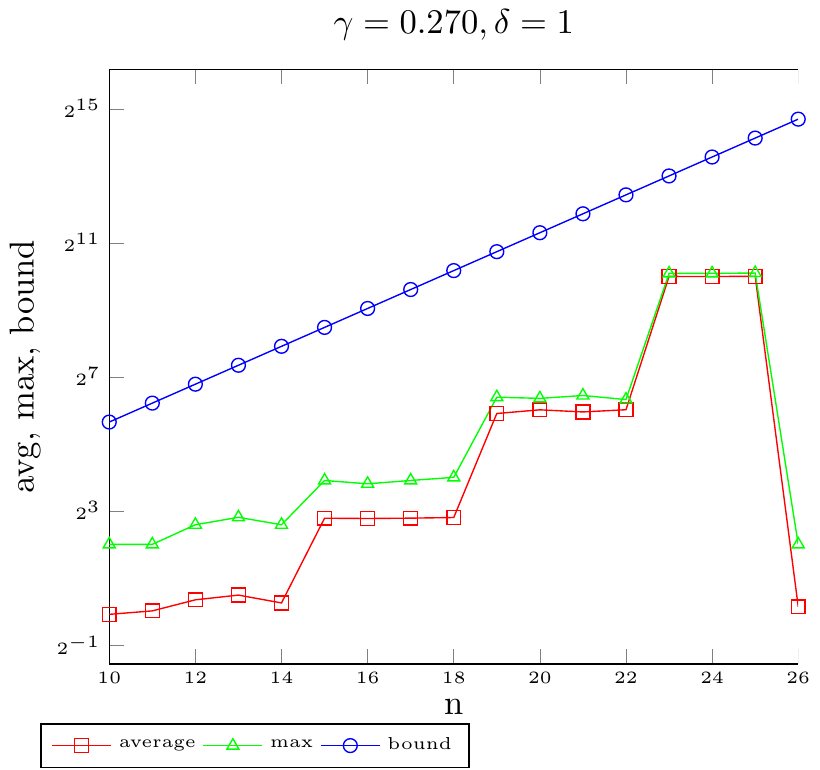}}}
\caption{Quality of the filtering with $\delta=1$.\label{fig:deltasup0}}
\end{figure}

Figure \ref{fig:deltasup0} shows further experimental evidence that adding redundancy by choosing $\delta=1$ permits to avoid problems occurring when the first nonpositive coefficient of $\mathsf{HS}_{n-k,m}$ is close to zero. For instance, the peak at $\gamma=.55$, $n=23$ that appeared in Figure \ref{fig:strongsemireg} disappears when $\delta=1$.

%% file: low_degree.tex
In some cases, when the boolean system is \emph{not random}, the
choice of $d_0$ proposed in Algorithm \textsf{BooleanSolve} may be
too large. This happens for instance for systems that have inner
structure, which has an impact on the algebraic structure of the ideal
generated by the polynomials. Examples of such structure can be found
in Cryptology, for instance with boolean systems coming for the HFE
cryptosystem \citep{Pat96}, as shown in~\citep{FauJou03}.

For these systems, the choice of $d_0$ as the index of the first
non-positive coefficient of $\mathsf{HS}_{n,m}$ would be very
pessimistic, since the Macaulay matrices in degree $d_0$ would be
larger than necessary. However, if estimates of the witness degree are
known (this is the case for HFE), then $d_0$ can be chosen accordingly
as a parameter of the Algorithm \textsf{BooleanSolve}.

%% file: crypto.tex
\subsection{Application in Cryptology}

Careful implementation of the algorithm will be necessary to
estimate accurately the efficiency of the {\sf BooleanSolve} algorithm.
For instance a crucial operation is the Wiedemann (or block Wiedemann)
algorithm; in practice, it is probably useless to work in a field
extension \(\mathbb{F}_{2^{k}}\) as requested by 
Proposition~\ref{thm:wiedemann}.
Working directly over \(\mathbb{F}_{2}\) and packing several
elements (bits) into words may have a dramatic effect on the
constant hidden in Theorem~\ref{thm:globalcompl}. In the following
we estimate the impact of the new algorithm from the point of
view of a user in Cryptology. In other words, if the security
of a cryptosystem relies on the hardness of solving a polynomial
system, by how much must the parameters be increased to keep the
same level of security?

The stream cipher QUAD~\citep{QUADEuro,QUADJSC} enjoys a provable security argument to support  its conjectured strength.  It relies on the iteration of a set of overdetermined multivariate quadratic polynomials over  \(\mathbb{F}_{2}\)
so that the security of the keystream generation is related, in the concrete security model, to the difficulty  of solving the Boolean MQ SAT\ problem. A theoretical bound is used in~\citep{QUADJSC}
to obtain secure parameters for a given security
bound \(T\) and a given maximal length $L$ of the keystream sequence that can be generated with a pair (key, IV): for instance (see~\cite{QUADJSC}
p. 1711), for \(T=2^{80}, L=2^{40}, k=2\) and an advantage of more than \(\varepsilon=1/100\), the bound gives
\(n\geq 331\). We report in the following table various values
of \(n\) depending on \(L\), \(T\) and \(\varepsilon\):

    \begin{center}
    \begin{tabular}{|c|c|c|c|}\hline
T & L & \(\varepsilon\) & n \\\hline
$2^{80}$ &   $2^{40}$ &   1/100    &       331\\\hline
$2^{80}$ &   $2^{22}$ &   1/100    &       253\\\hline
$2^{160}$ &  $2^{80}$ &   1/100    &       613\\\hline
$2^{160}$ &  $2^{40}$ &   1/100    &       445\\\hline
$2^{160}$ &  $2^{40}$ &   1/1000   &       448\\\hline
$2^{160}$ &  $2^{40}$ &   1/10000  &       467\\\hline
$2^{256}$ &  $2^{40}$ &   1/100    &       584\\\hline
$2^{256}$ &  $2^{80}$ &   1/100    &       758\\\hline
\end{tabular}\\
\vspace*{2mm}
Security parameters for the stream cipher QUAD~\citep{QUADJSC}
\end{center}

Now, the question is to achieve a security bound for \(T=2^{256}\); what are the minimal values
of \(m\) and \(n\) ensuring that solving the Boolean MQ SAT
requires at least \(T\) bit-operations? Using the complexity
analysis of the BooleanSolve algorithm we can derive useful
 lower bounds for \(n\) when \(m=n\) or \(m=2\,n\)
(\(m=2\,n\) corresponds to the recommended parameters for QUAD). In
the following table we report the corresponding values: 
 \begin{center}
 \begin{tabular}{|c|c|c|c|c|}\hline
Security Bound \(T\) & \rule{0cm}{5mm}$2^{128}$ & $2^{256}$ & $2^{512}$ & $2^{1024}$ \\\hline
Minimal value of \(n\) when \(m=n\)& \(128\) & \(270\) & \(576\) & \text{1202} \\\hline
Minimal value of \(n\) when \(m=2\,n\)& \(145\) & \(335\) & \(738\) & 1580 \\\hline
\end{tabular}
\end{center}

Comparing with exhaustive search we can see from this table that:
\begin{itemize}
\item our algorithm does not improve upon exhaustive search when \(n\)
  is small (for instance when \(m=n\) and \(T=2^{128}\) that are the
  recommended parameters);
\item by contrast, our algorithm can take advantage of the
  overdeterminedness of the algebraic systems: this explains why the
  values we recommend are larger than expected when \(n\) is large
  and/or \(m/n>1\).
\end{itemize}